\documentclass[runningheads]{llncs}
\usepackage[T1]{fontenc}
\usepackage{graphicx}
\usepackage{hyperref}

\hypersetup{
    colorlinks=true,
    linkcolor=blue,
    filecolor=blue,      
    urlcolor=blue,
    citecolor=blue,
    }

\usepackage{color}
\usepackage{url}

\usepackage{amsmath, amssymb}

\usepackage{tikz}
\usetikzlibrary{automata, positioning, arrows}
\tikzset{
  ->, % makes the edges directed
  >=stealth, % makes the arrow heads bold
  node distance=3cm, % specifies the minimum distance between two nodes. Change if necessary.
  initial text=$ $, % sets the text that appears on the start arrow
}
\usepackage{thmtools}
\usepackage[shortlabels]{enumitem}
\setlist[enumerate,1]{label={\arabic*.}}

\let\origendproof=\endproof
\def\endproof{\qed\origendproof}

\newcommand{\bb}{\mathbb}
\newcommand{\mc}{\mathcal}

\DeclareMathOperator{\dom}{dom}

\DeclareMathOperator{\Sym}{Sym}

\DeclareMathOperator{\Ldim}{Ldim}
\DeclareMathOperator{\Cdim}{Cdim}

\newif\ifappendix

\appendixtrue

\begin{document}
\title{Query Learning Bounds for Advice and Nominal Automata}
\author{Kevin Zhou}
% \author{Anonymized for submission}
%
\authorrunning{K. Zhou}
% \authorrunning{Anonymous}
\institute{University of Illinois Chicago, Chicago, IL 60607, USA \\
\email{kzhou23@uic.edu}\\
}

\maketitle              % typeset the header of the contribution
\begin{abstract}
Learning automata by queries is a long-studied area initiated by Angluin in 1987 with the introduction of the $L^*$ algorithm to learn regular languages, with a large body of work afterwards on many different variations and generalizations of DFAs.
Recently, Chase and Freitag introduced a novel approach to proving query learning bounds by computing combinatorial complexity measures for the classes in question, which they applied to the setting of DFAs to obtain qualitatively different results compared to the $L^*$ algorithm.
Using this approach, we prove new query learning bounds for two generalizations of DFAs.
The first setting is that of advice DFAs, which are DFAs augmented with an advice string that informs the DFA's transition behavior at each step.
For advice DFAs, we give the first known upper bounds for query complexity.
The second setting is that of nominal DFAs, which generalize DFAs to infinite alphabets which admit some structure via symmetries.
For nominal DFAs, we make qualitative improvements over prior results.

\keywords{Advice automata \and Nominal automata \and Query learning}
\end{abstract}
\section{Introduction}

Learning various forms of automata with queries is a well-studied class of problems with many applications, including in AI, automatic verification, and model checking. 
The field was initiated by Angluin in 1987 with the introduction of the $L^*$ algorithm that learns deterministic finite automata (DFAs) using a polynomially bounded number of equivalence and membership queries \cite{angluin:learning-regular-sets}.
The $L^*$ algorithm has seen many applications and has also inspired generalizations to other types of automata, such as tree automata \cite{yasubumi:learning-context-free-grammars}, nondeterministic finite automata \cite{bollig-etal:learning-NFA}, $\omega$-automata \cite{angluin-fisman:learning-regular-omega-languages}, symbolic automata \cite{drews-dantoni:learning-symoblic-automata}, and fully-ordered lattice automata \cite{fisman-saadon:learning-characterizing-fully-ordered-lattice-automata}. 
The setting of learning with queries is especially well-suited for learning automata---in real-world applications where the task is to identify the behavior of some black-box system that can be modelled by an automaton, queries can be simulated by interacting with the system and observing the output.

Historically, the study of query learning of automata has centered around adapting Angluin's $L^*$ algorithm to different settings.
We do not do this, and instead build on a method introduced by Chase \& Freitag \cite{chase-freitag:bounds-in-query-learning}, who give query learning bounds in terms of \emph{Littlestone dimension} and \emph{consistency dimension}, which are complexity measures for concept classes. 
One can then prove a query learning bound in a given setting by computing the Littlestone and consistency dimensions.
This approach seems to be particularly effective for learning automata, since many types of automata exhibit form of Myhill-Nerode characterization, which can be used to bound the consistency dimension.
For example, Chase and Freitag apply this method for regular languages and regular $\omega$-languages and obtain qualitatively different results from prior work.
Our work serves as further proof of concept of this method.

Our work studies the learnability of two generalizations of DFAs. 
The first is advice DFAs, which were studied as early as 1968 by Salomaa \cite{salomaa:finite-automata-time-variant-structure}, though we follow the notation of Kruckman et al. \cite{kruckman-etal:myhill-nerode-automata-advice}.
Advice DFAs generalize classical DFAs by allowing the automata access to an additional advice string that it reads concurrently with the input, and are useful in modeling situations where the transition behavior can vary over time. 
They are also of interest to the logic community: it is a classical result that DFAs correspond to weak monadic second-order formulas over the structure of natural numbers with the successor operation; advice DFAs correspond to formulas over expansions of this structure by unary predicates, a frequently studied setting; see e.g. \cite{elgot-rabin:decidability-undecidability-extensions-of-successor,carton-thomas:monadic-theory-morphic-infinite-words,rabinovich-thomas:decidable-theories-of-natural-numbers-unary-predicates,barany:automatic-omega-words-decidable-MSO-theory}.
Another motivating factor for advice DFAs comes from the study of \emph{automatic structures}, which are structures whose domain and atomic relations are recognized by DFAs. 
It turns out some natural structures are not automatic or even isomorphic to an automatic structure, such as $(\bb Q, +)$, the additive group of the rationals \cite{tsankov:additive-group-of-rationals-not-automatic}.
However, $(\bb Q, +)$ \textit{is} isomorphic to a structure whose domain and atomic relations are recognized by advice DFAs \cite{nies:describing-groups,kruckman-etal:myhill-nerode-automata-advice}.

The second setting we consider is that of nominal DFAs, introduced by Boja\'nczyk, Klin, and Lasota \cite{bojanczyk-klin-lasota:automata-theory-in-nominal-sets}, which generalize DFAs to infinite alphabets. 
Generalizing formal language theory to infinite alphabets is highly non-trivial, since without any restrictions, the fact that there are uncountably many subsets of any infinite set makes computation intractable.
However, in most reasonable applications, there is additional structure that can be leveraged to make computation reasonable, and nominal DFAs utilize the notion of \textit{nominal sets} (first introduced by Gabbay \& Pitts \cite{gabbay-pitts:new-approach-to-abstract-syntax}) to formalize this idea. 
Aside from the development of the theory of automata over infinite alphabets, nominal sets have also found much use in the concurrency and semantics communities as a formalism for modeling name binding (see e.g. \cite{montanari-pistore:history-dependent-automata-an-introduction,pitts:nominal-sets-names-and-symmetry-in-computer-science}).

\subsubsection{Organization of Paper and Summary of Results}

In \autoref{section:preliminaries}, we present the basic definitions and notation necessary for the remainder of the paper.

In \autoref{section:complexity of advice DFAs}, we study query learning of advice DFAs, and give the first known bound for the query complexity of advice DFAs.
Our result for advice DFAs is as follows: let $\mc L^{\text{adv}}_k(n,m)$ be the set of languages over an alphabet of size $k$ recognized by an advice DFA on at most $n$ states, restricted to strings of length at most $m$.
\begin{restatable*}{theorem}{adviceLC} \label{thm:EQ+MQ learning complexity of DFAs with advice}
    The (EQ+MQ)-query complexity of $\mc L^{\text{adv}}_k(n,m)$ with queries from $\mc L^{\text{adv}}_k(2n,m)$ is $O(n^3mk \log n)$.
\end{restatable*}

In \autoref{section:complexity of nominal DFAs}, we study query learning of nominal DFAs.
Query learning of nominal DFAs was previously studied by Moerman et al. \cite{moerman-etal:learning-nominal-automata}.
We give a qualitative improvement over this previous result; a detailed comparison is given in the discussion below.
Our result for nominal DFAs is as follows: given a $G$-alphabet $A$, let $\mc L^{\text{nom}}_A(n,k)$ denote the set of $G$-languages recognized by a nominal DFA with at most $n$ orbits and nominal dimension at most $k$.
\begin{restatable*}{theorem}{nominalLCfixedalphabet}
    For a fixed $G$-alphabet $A$, the (EQ+MQ)-query complexity of \\$\mc L^{\text{nom}}_A(n,k)$ with queries from $\mc L^{\text{nom}}_A(n,k)$ is at most $\frac{n^{O(k)}}{k^k}$.
\end{restatable*}

% Section \ref{section:conclusion} concludes the work. 

\ifappendix \else Due to space constraints, some proofs are omitted, and given in the appendix of the full version of this paper. \fi

\subsubsection{Discussion \& Related Work}

The approach of Chase \& Freitag of deriving query complexity bounds by computing bounds on various combinatorial notions of dimension is a fairly novel one, and in \cite{chase-freitag:bounds-in-query-learning} they illustrate its usefulness by easily proving polynomial query learning bounds for regular languages. 
The approach gives a qualitatively different bound compared to the one given by the $L^*$ algorithm---their results have no dependence on the length of the longest counterexample returned by the oracle (while the $L^*$ algorithm does), but a worse dependence on the number of states. 
They also do not give any computational complexity guarantees---their algorithm requires computing the Littlestone dimension and finding witnesses to the consistency dimension, which in general may be hard computational problems.

As mentioned, in the setting of advice DFAs, we give the first known bounds on the query complexity. 
The restriction on the length of the strings in $\mc L^{\text{adv}}_k(n,m)$ is necessary to make the problem tractable and can be thought of as giving a dependence on the length of the longest counterexample for equivalence queries.
It may be of interest to see if there is an analogous version of the $L^*$ algorithm for advice DFAs, and if so, how the bounds with that approach compare to ours.

The prior best known bound for query complexity of nominal DFAs is given by Moerman et al. in \cite{moerman-etal:learning-nominal-automata}, using a generalization of Angluin's $L^*$ algorithm to nominal DFAs. The bound that they derive is a complicated quantity (see \cite[Corollary 1]{moerman-etal:learning-nominal-automata}), but in particular we observe the following: if the target automaton has $n$ orbits and nominal dimension $k$, $p$ is the nominal dimension of the (fixed) alphabet, and $m$ is the length of the longest counterexample returned by the oracle, then their bound is lower bounded by both
\begin{enumerate}[(a)]
    \item $\min\left(\left(\frac{nk}{e}\right)^m, \left(\frac{m}{e}\right)^{nk}\right)$ and 
    \item $(k^n n!)^p$.
\end{enumerate}
These are not explicitly stated in the original work, but follow with some additional work (see \autoref{remark:factorial bounds in moerman etal}).

Our result improves on (a) since our bound does not depend on the length of the longest counterexample, and improves the asymptotic dependence on $n$ compared to (b) (polynomial in $n$ with respect to $k$, instead of factorial in $n$).
This is especially important in light of \autoref{cor:dimension of A^*/equiv_L bounded in terms of orbits}, which says that in this setting, $k$ is at most a constant multiple of $n$.
However, as with Chase \& Freitag's result on classical DFAs, we give no computational complexity guarantees.

\section{Preliminaries} \label{section:preliminaries}

\subsection{Query learning}

Let $X$ be a set (the \emph{instance space}). 
A \emph{concept} is a function $C : X \to \{0,1\}$, and a \emph{concept class} $\mc C$ on $X$ is a nonempty set of concepts. 
We note that a concept is sometimes equivalently defined as a subset of $X$, but for our purposes it will be easier to work with functions. 
Fix a concept $C \in \mc C$, which we call the \textit{target}, and another concept class $\mc H \supseteq \mc C$, which we call the \emph{hypothesis class}. 
An \emph{equivalence query} (EQ) consists of a hypothesis $H \in \mc H$, to which the oracle answers \textit{yes} if $H = C$, or with a counterexample $x \in X$ for which $H(x) \neq C(x)$. 
A \emph{membership query} (MQ) consists of an element $x \in X$, to which the oracle responds with the value of $C(x)$. 
The learning procedure proceeds interactively in rounds: in each round, the learner poses a query to the oracle, and the oracle responds with the corresponding answer. 
The learner is allowed to choose queries based on the responses to previous queries, and succeeds if they submit the target as an equivalence query. 
In \emph{EQ-learning}, the learner is only allowed to submit equivalence queries, while in \emph{(EQ+MQ)-learning}, the learner can use both equivalence and membership queries. 

\begin{definition}[Query Complexity]
    Let $\mc C \subseteq \mc H$ be two concept classes. 
    The \emph{EQ-query complexity of $\mc C$ with queries from $\mc H$} is defined to be the least $n$ such that there is an algorithm for the learner to submit equivalence queries from $\mc H$ with the property that for any $C \in \mc C$, the learner can identify $C$ within at most $n$ queries, or $\infty$ if no such $n$ exists.
    
    The \emph{(EQ+MQ)-query complexity of $\mc C$ with queries from $\mc H$} is defined in the same way, except that the learner is allowed to also use membership queries.
\end{definition}

A significant stream of prior work studied bounds for the query complexity of $\mc C$ with hypotheses from $\mc H$ in terms of combinatorial complexity measures of $\mc C$ and $\mc H$ \cite{hellerstein-pillaipakkamnatt-raghavan-wilkins:how-many-queries-are-needed-to-learn,balcazar-castro-guijarro-simon:consistency-dimension-query-learning,chase-freitag:bounds-in-query-learning,hanneke-livni-moran:online-learning-simple-predictors}.
These bounds are formulated in terms of the \emph{Littlestone dimension}, \emph{consistency dimension}, and \emph{strong consistency dimension} (also known as the \emph{dual Helly number}), which we define now.

A \emph{binary element tree} is a complete binary tree whose internal nodes are labeled by elements of $X$. 
A binary element tree $T$ is \emph{shattered} by $\mc C$ if there is a way to label all the leaves of $T$ with elements of $\mc C$ such that the following condition holds: given a leaf node labeled by $A \in \mc C$, for each internal node above $A$ labeled by $x \in X$, $A(x)=1$ if and only if the (unique) path from the root to $A$ goes through the left child of $x$. 

An example of a (labelled) binary element tree of height 2 is given below, where $x_0, x_1, x_2 \in X$, and $C_0, C_1, C_2, C_3 \in \mc C$:

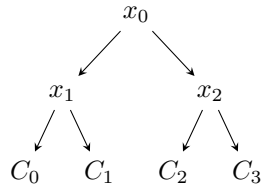
\begin{center}
\begin{tikzpicture}[level/.style={sibling distance=8em/#1},scale=0.7]
  \node {$x_0$}
    child {node {$x_1$}
      child {node {$C_0$}}
      child {node {$C_1$}}}
    child {node {$x_2$}
      child {node {$C_2$}}
      child {node {$C_3$}}}
    ;
\end{tikzpicture}
\end{center}

This tree is shattered exactly when $C_0(x_0) = C_0(x_1) = 1$, $C_1(x_0) = 1$ but $C_1(x_1) = 0$, $C_2(x_0) = 0$ but $C_2(x_2) = 1$, and $C_3(x_0) = C_3(x_2) = 0$.

\begin{definition}[Littlestone dimension]
    The \emph{Littlestone dimension} of a concept class $\mc C$, denoted $\Ldim(\mc C)$, is the maximum $n$ such that there exists a binary element tree $T$ of height $n$ which is shattered by $\mc C$. 
    If no such $n$ exists, we say that $\Ldim(\mc C) = \infty$. 
\end{definition}

\begin{remark} \label{remark:littlestone dimension of a finite class}
    A straightforward bound for the Littlestone dimension of a finite class $\mc C$ is $\Ldim(\mc C) \leq \log |\mc C|$. To see this, note that a binary element tree $T$ of height $>\log |\mc C|$ has more than $|\mc C|$ leaves, so there must be two leaves labeled by the same element $C \in \mc C$. 
    Consider the internal node where the paths leading to these two leaves differ; suppose it is labeled by $x$. 
    Then there is a root-to-leaf path ending at $C$ that goes through the left child of $x$, implying that $x \in C$.
    However, there is also a root-to-leaf path ending at $C$ that goes through the right child of $x$, implying that $x \notin C$, so $T$ cannot be shattered.
\end{remark}

We now define consistency dimension.
For the following bulleted definitions, let $A,B$ be \underline{partial} functions from $X$ to $\{0,1\}$.
\begin{itemize}
    \item $\dom(A)$ denotes domain of $A$.
    \item The \emph{size} of $A$ refers to the cardinality of $\dom(A)$.
    \item For a set $Y \subseteq \dom(A)$, the \emph{restriction} of $A$ to $Y$ is the partial function $A|_Y$ defined by $A|_Y(x) = A(x)$ for $x \in Y$ and undefined outside of $Y$.
    \item We say that $B$ \emph{extends} $A$ if $\dom(A) \subseteq \dom(B)$ and $B|_{\dom(A)} = A$. 
    On the other hand, we say that $A$ is a \emph{restriction} of $B$. 
    \item Given a concept class $\mc C$, $A$ is \emph{$n$-consistent with $\mc C$} if every size $n$ restriction of $A$ has an extension in $\mc C$. 
    Otherwise, $A$ is \emph{$n$-inconsistent}. 
\end{itemize}

\begin{definition}[Consistency Dimension]
    The \emph{consistency dimension} of $\mc C$ with respect to $\mc H$, denoted $\Cdim(\mc C, \mc H)$, is the least $n$ such that for every concept $A : X \to \{0,1\}$ that is $n$-consistent with $\mc C$, we have that $A \in \mc H$. 
    If no such $n$ exists, we say that $\Cdim(\mc C, \mc H) = \infty$. 
    In the case that $\mc H = \mc C$, we will write $\Cdim(\mc C)$ to denote $\Cdim(\mc C, \mc C)$. 
\end{definition}

\begin{remark}
    By contrapositive, $\Cdim(\mc C, \mc H) \leq n$ if for every concept $A \notin \mc H$, $A$ is $n$-inconsistent with $\mc C$. That is, we can find a subset $Y \subseteq X$ of size $n$ such that $A|_Y$ cannot be extended to anything in $\mc C$. 
\end{remark}

It is often easier to work concretely with the contrapositive characterization of consistency dimension. 
Our proofs of bounds on consistency dimension will go through this direction.

\begin{theorem}[{\cite[Theorem 2.24]{chase-freitag:bounds-in-query-learning}}]
\label{thm:query complexity bounds}
Let $\mc C \subseteq \mc H$ be two concept classes on a set $X$, $c = \Cdim(\mc C, \mc H)$, and $d = \Ldim(\mc C)$.
Then (EQ+MQ)-query complexity of $\mc C$ with queries from $\mc H$ is $O(cd)$. 
\end{theorem}

% Note that item 3 may seem to improve the dependence of EQ-query complexity on Littlestone dimension from exponential to linear.
% However, as seen above, the move from $\Cdim$ to $\SCdim$ can reintroduce an exponential dependence on $\Ldim$, so it is not guaranteed to be a stronger bound.

\subsection{Deterministic finite automata}

In this subsection, we set notation and terminology for deterministic finite automata and regular languages, and briefly review the Myhill-Nerode theorem.

\begin{definition}[Deterministic Finite Automata]
    Let $\Sigma$ be a finite set (the \emph{input alphabet}). 
    A \emph{deterministic finite automaton (DFA) over $\Sigma$} consists of the following data:
    \begin{itemize}
        \item a finite set $Q$ (the set of \emph{states}); 
        \item a function $\delta : Q \times \Sigma \to Q$ (the \emph{transition function});
        \item a state $q_0 \in Q$  (the \emph{initial state}); and
        \item a set of states $F \subseteq Q$  (the set of \emph{accepting states}).
    \end{itemize}
\end{definition}

As is typical, $\Sigma^*$ will denote the set of finite strings over $\Sigma$.
Given an input string $x = x_1 x_2 \cdots x_n \in \Sigma^*$, and a DFA $M$, define the \emph{run} of $M$ on $x$ to be the sequence of states $\alpha_0, \ldots, \alpha_n \in Q$ such that $\alpha_0 = q_0$, and for $1 \leq i \leq n$, we have that $\delta(\alpha_{i-1}, x_i) = \alpha_i$. 
A string $x \in \Sigma^*$ is \emph{accepted} by $M$ if the last state appearing in the run of $M$ on $x$ is in $F$.
A \emph{language} is a function $L : \Sigma^* \to \{0,1\}$.
We note that languages are often defined as subsets of $\Sigma^*$, but here we use functions in order to align with our definition of a concept.
A language is \emph{recognized} by a DFA $M$ if $M$ accepts a string $x$ if and only if $x \in L$.
We say that a language $L$ is \emph{regular} if it is recognized by some DFA.

The Myhill-Nerode theorem provides a useful syntactic characterization of regular languages.
Given a language $L : \Sigma^* \to \{0,1\}$, define an equivalence relation $\equiv_L$ on $\Sigma^*$ as follows: for strings $x,y \in \Sigma^*$, we say that $x \equiv_L y$ iff $L(xz)= L(yz)$ for all $z \in \Sigma^*$.

\begin{theorem}[Myhill-Nerode] \label{thm:myhill-nerode}
  A language $L$ is regular if and only if $\equiv_L$ has finitely many equivalence classes.
  Moreover, $\equiv_L$ has exactly $n$ classes if and only if the minimal DFA recognizing $L$ has exactly $n$ states.
\end{theorem}

\section{Learning Advice DFAs} \label{section:complexity of advice DFAs}

\subsection{Overview of advice DFAs}

In this subsection, we give an introduction to automata with advice.
For a more comprehensive overview, see \cite{kruckman-etal:myhill-nerode-automata-advice}.

\begin{definition}[Advice DFA]
  Let $\Sigma$ and $\Gamma$ be finite sets (the \emph{input} and \emph{advice alphabets}, respectively). 
  An \emph{advice DFA $M$ over $\Sigma$ with advice from $\Gamma$} consists of the following data: 
  \begin{itemize}
        \item a finite set $Q$ (the set of \emph{states}); 
        \item an infinite length string $A \in \Gamma^\omega$ over the advice alphabet (the \emph{advice string}); 
        \item a function $\delta : Q \times \Sigma \times \Gamma \to Q$ (the \emph{transition function});
        \item a state $q_0 \in Q$  (the \emph{initial state}); and
        \item a set of states $F \subseteq Q$  (the set of \emph{accepting states}).
  \end{itemize}
\end{definition}

Given an input string $x = x_1 x_2 \cdots x_n \in \Sigma^*$ and advice DFA $M$, define the \emph{run} of $M$ on $x$ to be the sequence of states $\alpha_0, \ldots, \alpha_n \in Q$ such that $\alpha_0 = q_0$, and for $1 \leq i \leq n$, we have that $\delta(\alpha_{i-1}, x_i, A_i) = \alpha_i$. 
A string $x \in \Sigma^*$ is accepted by $M$ if the last state appearing in the run of $M$ on $x$ is in $F$.
A language $L$ is \emph{recognized} by an advice DFA $M$ if $M$ accepts a string $x$ if and only if $x \in L$.
We say that a language $L$ is \emph{regular with advice} if it is recognized by some advice DFA.

Advice DFAs extend classical DFAs with the addition of the (infinite) advice string $A$ (which is fixed beforehand as part of the automaton). 
The advice string is read in parallel with the input string; i.e., when $M$ reads the $n^\text{th}$ character of the input, it also has access to the $n^\text{th}$ character of the advice when deciding which transition to make.
One way to think about the advice string is that it allows the transition function to vary at each step of the computation (although at a fixed step $i$, the transition behavior is the same regardless of the input string).

Advice DFAs satisfy a Myhill-Nerode characterization, under a variant of the $\equiv_L$ relation.
Define an equivalence relation $\equiv_{L,m}$ on $\Sigma^m$ by $x \equiv_{L,m} y$ iff $xz \in L \iff yz \in L$ for all $z \in \Sigma^*$. Notice that $\equiv_{L,m}$ is simply $\equiv_L$ restricted to strings of length $m$. 

\begin{theorem}[Myhill-Nerode for advice DFAs, cf. {\cite[Theorem 4]{kruckman-etal:myhill-nerode-automata-advice}}] \label{thm:advice myhill-nerode with bounds}
    Let $L$ be a language.
    \begin{enumerate}[(i)]
        \item Suppose $L$ is accepted by an advice DFA that has $n$ states. 
        Then $\equiv_{L,m}$ has at most $n$ classes for all $m \in \bb N$. 
        
        \item Suppose $\equiv_{L,m}$ has at most $n$ classes for all $m \in \bb N$. 
        Then there is an advice DFA on $2n$ states that recognizes $L$.
    \end{enumerate}
\end{theorem}

In general, the bound of $2n$ states in the second item is tight.
\ifappendix A example witnesing this is given in Appendix \ref{appendix:witness to tightness of advice myhill-nerode bounds}. \fi
Also, note that this statement is more precise than the original version in \cite{kruckman-etal:myhill-nerode-automata-advice}; in particular, the relationship between the number of states and the number of $\equiv_{L,m}$-classes does not appear in the original theorem. 
However, this relationship is easily derived from its proof.

\subsection{Learning bound for advice DFAs}

For the remainder of the section, fix $k \in \bb N$ and let $\Sigma$ be an alphabet of size $k$. 
We will not fix the advice alphabet $\Gamma$; however, notice that if we are constructing an automaton with at most $n$ states, we may take $\Gamma$ to have size $n^{nk}$, since the advice string can be thought of as coding the transition function at a given step, and there are $n^{nk}$ functions from $Q \times \Sigma \to Q$ (possible transition functions). 

Note that the language consisting of the finite prefixes of any $\omega$-word $A$ is regular with advice $A$, and each of these languages is distinct, so there are uncountably many languages that are regular with advice.
Thus we cannot have any finitary representation of arbitrary languages that are regular with advice.
Hence, we consider the case where we restrict to strings of bounded length.
Let $\mc L^{\text{adv}}_k(n,m)$ denote the set 
\begin{align*}
    \mc L^{\text{adv}}_k(n,m) := \{L \subseteq \Sigma^{\leq m} \ \mid \ &L \text{ is recognized by an advice} \\
    &\text{DFA with at most } n \text{ states}\},
\end{align*}
and let $\mc E_k(n,m)$ denote the set
\begin{align*}
  \mc E_k(n,m) := \{L \subseteq \Sigma^{\leq m}  \ \mid \ &\equiv_{L,\ell} \text{ has at most } \\
  &n \text{ classes for all } \ell \leq m\}.
\end{align*}

By \autoref{thm:advice myhill-nerode with bounds}, $\mc L^{\text{adv}}_k(n,m) \subseteq \mc E_k(n,m) \subseteq \mc L^{\text{adv}}_k(2n,m)$ for any $n,m \in \bb N$. 
Because of this, it is more convenient to compute the query complexity of $\mc L^{\text{adv}}_k(n,m)$ with queries $\mc L^{\text{adv}}_k(2n,m)$. 

\begin{proposition} \label{prop:Cdim of DFAs with advice}
  The consistency dimension of $\mc L^{\text{adv}}_k(n,m)$ with respect to \\$\mc L^{\text{adv}}_k(2n,m)$ is at most $n(n+1)$.
\end{proposition}

\begin{proof}
    Let $L : \Sigma^{\leq m} \to \{0,1\}$, and suppose that $L \notin \mc L^{\text{adv}}_k(2n,m)$. 
    Since $\mc E_k(n,m) \subseteq \mc L^{\text{adv}}_k(2n,m)$, we have that $L \notin \mc E_k(n,m)$.
    This means that there are strings $x_0, \ldots, x_n \in \Sigma^\ell$ which are pairwise $\equiv_{L,\ell}$-inequivalent for some $\ell \leq m$. 
    For each $0 \leq i < j \leq n$, let $z_{ij} \in \Sigma^*$ such that $L(x_i z_{ij}) \neq L(x_j z_{ij})$ (i.e., $z_{ij}$ distinguishes $x_i$ and $x_j$ according to $\equiv_{L,\ell}$). 
    Consider the set 
    \[
        B = \{x_k z_{ij} \mid 0 \leq i < j \leq n, k = i,j\},
    \]
    which has size $2 \binom{n+1}{2} = n(n+1)$. 
    Let $L' : \Sigma^{\leq m} \to \{0,1\}$ extend $L|_B$.
    $L'$ must have at least $n+1$ $\equiv_{L',\ell}$-classes, since $x_0, \ldots, x_n$ are forced to be $\equiv_{L',\ell}$-inequivalent. 
    Hence $L' \notin \mc E_k(n,m)$, so in particular $L' \notin \mc L^{\text{adv}}_k(n,m)$.
    Since we found a set of size $n(n+1)$ which witnesses the fact that $L$ is $n$-inconsistent with $\mc L^{\text{adv}}_k(n,m)$, the consistency dimension of $\mc L^{\text{adv}}_k(n,m)$ with respect to $\mc L^{\text{adv}}_k(2n,m)$ is at most $n(n+1)$. 
\end{proof}

\begin{proposition} \label{prop:Ldim of DFAs with advice}
    The Littlestone dimension of $\mc L^{\text{adv}}_k(n,m)$ is $O(nmk \log n)$. 
\end{proposition}

\begin{proof}
    We can bound the Littlestone dimension of $\mc L^{\text{adv}}_k(n,m)$ by bounding the size of $\mc L^{\text{adv}}_k(n,m)$ and applying \autoref{remark:littlestone dimension of a finite class}.
    As noted at the beginning of the section, we may interpret $\Gamma$ as coding all possible transition functions, so the advice string $A$ simply tells the automaton which transition function to use at each step. 
    
    Consider an advice DFA such that $Q = [n]$ and $q_0 = 1$. 
    Then to fully specify the behavior of the automaton on strings of length $m$, it is enough to choose one transition function for each step, as well as the set of accepting states. 
    There are $(n^{nk})^m = n^{nmk}$ ways to choose the transition functions, and $2^n$ ways to choose the set of accepting states. 
    So there are $n^{nmk} 2^n$ total possible advice DFAs of this form. 
    
    Now notice that every language on strings of length up to $m$ that is accepted by an advice DFA with at most $n$ states is accepted by an advice DFA of this form, simply by appropriately permuting the states and the letters of the advice string. 
    So indeed the size of the class $\mc L^{\text{adv}}_k(n,m)$ is upper bounded by $n^{nmk} 2^n$.
    Therefore
    \begin{align*}
    \Ldim(\mc L^{\text{adv}}_k(n,m)) &\leq \log |\mc L^{\text{adv}}_k(n,m)| \\
    &= \log(n^{nmk} 2^n) \\
    &= nmk \log n + n \\
    &= O(nmk \log n).
    \end{align*}
\end{proof}

Combining Propositions \ref{prop:Cdim of DFAs with advice} and \ref{prop:Ldim of DFAs with advice} with \autoref{thm:query complexity bounds}(1), we obtain:

\adviceLC

\section{Learning Nominal DFAs} \label{section:complexity of nominal DFAs}

\subsection{Overview of nominal DFAs}

In this subsection, we give an introduction to nominal automata.
For a more comprehensive treatment, see \cite{bojanczyk-klin-lasota:automata-theory-in-nominal-sets}.

As mentioned in the introduction, one must leverage some underlying structure to properly generalize automata theory to infinite alphabets.
For example, consider the alphabet $A = \bb N$, and define $L : A^* \to \{0,1\}$ by $L(w) = 1$ if and only if $w = aa$ for some $a \in A$.
An equivalently defined language over a finite alphabet is easily seen to be regular, and so we expect this example to be ``regular'' as well. We can draw an infinite automaton that recognizes $L$ in the following way:

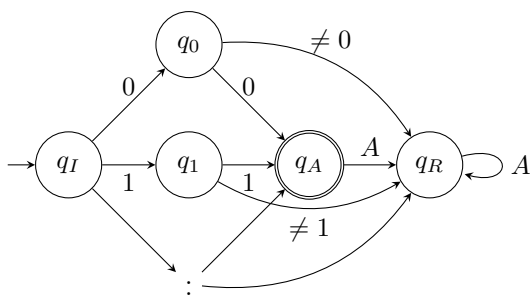
\begin{center}
\begin{tikzpicture}[scale=.8]
  \node[state, initial] (qI) at (0,0) {$q_I$};
  \node[state] (q0) at (2,2) {$q_0$};
  \node[state] (q1) at (2,0) {$q_1$};
  \node (vdots) at (2,-2) {$\vdots$};
  \node[state, accepting] (qA) at (4,0) {$q_A$};
  \node[state] (qR) at (6,0) {$q_R$};

  \draw[->] (qI) edge[above] node{$0$} (q0)
  (qI) edge[below] node{$1$} (q1)
  (qI) edge (vdots)
  (q0) edge[above] node{$0$} (qA)
  (q1) edge[below] node{$1$} (qA)
  (vdots) edge (qA)
  (q0) edge[above, bend left] node{$\neq 0$} (qR)
  (q1) edge[below, bend right] node{$\neq 1$} (qR)
  (vdots) edge[bend right] (qR)
  (qA) edge[above] node{$A$} (qR)
  (qR) edge[loop right] node{$A$} (qR)
  ;
\end{tikzpicture}
\end{center}

Notice that in some cases, we have combined infinitely many transitions into a single arrow---here, we assume that we are able to compare whether or not two data values are equal, and so are leveraging a form of symmetry.
We can further condense this representation into an actually finite diagram:

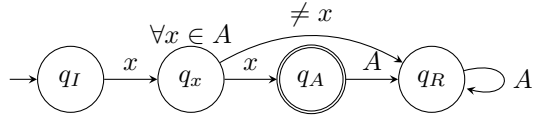
\begin{center}
\begin{tikzpicture}[scale=.8]
  \node[state, initial] (qI) at (0,0) {$q_I$};
  \node[state] (qx) at (2,0) {$q_x$};
  \node[state, accepting] (qA) at (4,0) {$q_A$};
  \node[state] (qR) at (6,0) {$q_R$};

  \node at (2,.75) {$\forall x \in A$};

  \draw[->] (qI) edge[above] node{$x$} (qx)
  (qx) edge[above] node{$x$} (qA)
  (qx) edge[above, bend left] node{$\neq x$} (qR)
  (qA) edge[above] node{$A$} (qR)
  (qR) edge[loop right] node{$A$} (qR)
  ;
\end{tikzpicture}
\end{center}

We can compress the distinct states for each character from $A$ into a single ``state'' because it is enough to be able to compare whether or not the first and second characters read are equal or not.
In order to formalize this intuition, we use the notion of \emph{nominal sets}.
Note: we do not work in the most general setting of \cite{bojanczyk-klin-lasota:automata-theory-in-nominal-sets}, and instead only focus on what they call the \textit{equality symmetry}.
However, it is reasonable to expect that our results can generalize to other well-behaved symmetries.

Given a set $A$, let $\Sym(A)$ denote the group of permutations on $A$, i.e., the set of bijections from $A$ to $A$ with the operation of function composition.
For the rest of the paper, let $G$ denote $\Sym(\bb N)$
Given a set $X$, a \emph{(left) action} of $G$ on $X$ is an operation $\cdot : G \times X \to X$ such that:
\begin{enumerate}[1)]
    \item for all $x \in X$, $e \cdot x = x$, where $e$ is the identity function on $\bb N$, and 
    \item for all $\pi,\pi' \in G$ and $x \in X$, $(\pi \circ \pi') \cdot x = \pi \cdot (\pi' \cdot x)$.
\end{enumerate}

\begin{example} \label{eg:examples of group actions}
    Let $X = \bb N^2$. 
    Define an action of $G$ on $X$ by $\pi \cdot (n,m) = (\pi(n), \pi(m))$.
    Observe that there is a permutation $\pi$ for which $\pi \cdot (n_1, m_1) = (n_2, m_2)$ if and only if either (1) $n_1 = m_1$ and $n_2 = m_2$, or (2) $n_1 \neq m_1$ and $n_2 \neq m_2$. 
    This example illustrates the action of $\Sym(\bb N)$ can be used to formalize the idea of being able to compare data values for equality.
\end{example}

\begin{definition}[Support; Least Support]
    Fix $X$ and an action of $G$ on $X$.
    Given a subset $D \subseteq \bb N$ and an element $x \in X$, we say that $D$ \emph{supports} $x$ if for every $\pi \in G$ such that $\pi$ fixes every element of $D$, we have that $\pi \cdot x = x$.

    A finite set $D \subseteq \bb N$ is called the \emph{least support of $x$} if 
    \begin{enumerate}[(1)]
        \item $D$ supports $x$, 
        \item no proper subset of $D$ supports $x$, and 
        \item no other finite subset of $\bb N$ has properties (1) and (2).
    \end{enumerate}
    We will use $supp(x)$ to denote the least support of $x$.
\end{definition}

An equivalent characterization of an element having least support is that the intersection of any two finite supports of $x$ also supports $x$.

\begin{definition}[Nominal Set]
    A \emph{nominal set} is a set $X$ along with an action of $G$ on $X$ such that every element of $X$ has a least support.
\end{definition}

\begin{lemma}[{\cite[Lemma 4.9]{bojanczyk-klin-lasota:automata-theory-in-nominal-sets}}] \label{lem:size of support preserved under action}
    Let $X$ be a nominal set and let $x \in X$.
    If $D$ supports $x$, then for any $\pi \in G$, $\pi(D)$ supports $\pi \cdot x$.

    As a corollary, $|supp(x)| = |supp(\pi \cdot x)|$ for any $\pi \in G$.
\end{lemma}

\begin{definition}[Nominal Dimension]
    Let $X$ be a nominal set. 
    The \emph{nominal dimension} of $X$ is $\displaystyle \sup_{x \in X} |supp(x)|$ (i.e., the largest size of a least support).
\end{definition}

In the literature on nominal sets, this is simply referred to as the dimension of $X$. 
However, in order to prevent confusion with other notions of dimension used in this paper, we will use the term nominal dimension.

\begin{definition}[Orbit; Orbit-finiteness]
    Let $X$ be a nominal set.
    The \emph{orbit} of an element $x \in X$, denoted $G \cdot x$, is the set
    \[
    G \cdot x = \{\pi \cdot x \mid \pi \in G\}.
    \]
    Every nominal set is partitioned into the disjoint union of its orbits.
    We say that $X$ is \emph{orbit-finite} if $X$ is the union of only finitely many orbits. 
\end{definition}

By \autoref{lem:size of support preserved under action}, elements in the same orbit have the same size of least support, so orbit-finite nominal sets have finite nominal dimension.

\begin{example}
    Let $X = \bb N^2$, and define the action of $G$ on $X$ as in \autoref{eg:examples of group actions}.
    The orbits of $X$ are $G \cdot (n,n)$ and $G \cdot (n,m)$, where $n,m \in \bb N$ and $n \neq m$.
    These are the sets of pairs whose coordinate are either equal or unequal, respectively.
    Thus $X$ is orbit-finite with two orbits.
    Additionally, every element $(n,m)$ has least support $\{n,m\}$, and so the nominal dimension of $X$ is 2.
\end{example}

For $i=1, \ldots, n$, let $X_i$ be a set with an action of $G$ on $X_i$. 
The \emph{pointwise action} of $G$ on $\prod_{i=1}^n X_i$ is the action given by $\pi \cdot (x_1, \ldots, x_n) = (\pi \cdot x_1, \ldots, \pi \cdot x_n)$.

\begin{definition}[Equivariance]
    Let $X$ be a nominal set.
    A subset $Y \subseteq X$ is \emph{equivariant} if for any $y \in Y$ and $\pi \in G$, $\pi \cdot y \in Y$. 
    
    A relation $R$ on $\prod_{i=1}^n X_i$, where each $X_i$ is a nominal set, is \emph{equivariant} if it is equivariant when considered as a subset of the product equipped with the pointwise action. 
    
    In particular, a function $f : X \to Y$ is equivariant exactly when
    \[
    f(\pi \cdot x) = \pi \cdot f(x).
    \]
\end{definition}

\begin{remark}
    $Y$ is an equivariant subset of $X$ if and only if $Y$ is a union of orbits.
\end{remark}

The following lemmas demonstrate how equivariant functions behave nicely with orbits and supports.

\begin{lemma} \label{lem:orbit of image under equivariant function}
    Let $f : X \to Y$ be an equivariant function, and let $x \in X$. 
    The image of the orbit of $x$ (in $X$) under $f$ is equal to the orbit of $f(x)$ (in $Y$).

    As a corollary, if $f$ is surjective, then $Y$ has at most as many orbits as $X$.
\end{lemma}

\begin{proof}
    The orbit of $f(x)$ is the set $\{\pi \cdot f(x) \mid \pi \in G\}$.
    By equivariance of $f$, this is equal to $\{f(\pi \cdot x) \mid \pi \in G\} = f(\{\pi \cdot x \mid \pi \in G\})$, which is the image under $f$ of the orbit of $x$.
\end{proof}

\begin{lemma}[{\cite[Lemma 4.8]{bojanczyk-klin-lasota:automata-theory-in-nominal-sets}}] \label{lem:supports preserved under equivariant functions}
    Let $f : X \to Y$ be an equivariant function, $x \in X$, and $D \subseteq \bb N$.
    If $D$ supports $x$, then $D$ supports $f(x)$. 

    As a corollary, $supp(f(x)) \subseteq supp(x)$, and hence if $f$ is surjective, the nominal dimension of $Y$ is at most the nominal dimension of $X$.
\end{lemma}

We will work extensively with quotients by equivariant equivalence relations, so we state some important general facts about them.
Recall that if $R$ is an equivalence relation on $X$, $X/R$ denotes the set of equivalence classes of $R$, and is called the \emph{quotient} of $X$ by $R$.

\begin{lemma}[{\cite[Lemma 3.5]{bojanczyk-klin-lasota:automata-theory-in-nominal-sets}}] \label{lem:quotient by equivariant equivalence relation is nominal}
    Let $X$ be a nominal set and $R \subseteq X \times X$ be an equivariant equivalence relation. 
    Then the quotient $X / R$ is a nominal set, under the action $\pi \cdot [x]_R = [\pi \cdot x]_R$, and the quotient map $x \mapsto [x]_R$ is a surjective equivariant function.
\end{lemma}

\begin{lemma} \label{lem:dimension of quotient at most dimension of original}
    Let $X$ be a nominal set and $R \subseteq X \times X$ be an equivariant equivalence relation.
    Then $supp([x]_R) \subseteq supp(x)$ for every $x \in X$, and the nominal dimension of $X/R$ is at most the nominal dimension of $X$.
\end{lemma}
\begin{proof}
    This follows immediately from \autoref{lem:supports preserved under equivariant functions} and the fact that the quotient map $x \mapsto [x]_R$ is surjective and equivariant.
\end{proof}

\begin{lemma} \label{lem:induced equivariant equivalence relations}
    Let $X, Y$ be nominal sets.
    An equivariant function $F : X \to Y$ and equivariant equivalence relation $\equiv_Y$ on $Y$ induce an equivariant equivalence relation $\equiv_X$ on $X$ and an equivariant function $f : X / \equiv_X \to Y / \equiv_Y$.

    Furthermore, $f$ is injective, and if $F$ is surjective, then $f$ is also surjective.
\end{lemma}

\ifappendix A proof of \autoref{lem:induced equivariant equivalence relations} is given in Appendix \ref{appendix:proof of induced equivariant equivalence relations}. \fi

We are now ready to define nominal DFAs and state the nominal Myhill-Nerode theorem.
Fix an orbit-finite nominal set $A$, which we will call the \textit{$G$-alphabet}. 
The action of $G$ on $A$ naturally extends to $A^*$: given a string $w \in A^*$, $\pi \cdot w$ is the string obtained by letting $\pi$ act on each individual character of $w$.  

\begin{definition}[$G$-language]
    A \emph{$G$-language} is a function $L : A^* \to \{0,1\}$ such that the set $\{x \in A^* \mid L(x) = 1\}$ is an equivariant subset of $A^*$. 
\end{definition}

\begin{definition}[Nominal DFA]
    Let $A$ be an orbit-finite nominal set (the \emph{input alphabet}).
    A \emph{nominal DFA $M$ over $A$} consists of the following data:
    \begin{itemize}
        \item an orbit-finite nominal set $Q$ (the set of \emph{states});
        \item an equivariant function $\delta : Q \times A \to Q$ (the \emph{transition function})
        \item a state $q_0 \in Q$ such that the orbit of $q_0$ is $\{q_0\}$ (the \emph{initial state})
        \item an equivariant subset $F \subseteq Q$ (the set of \emph{accepting states})
    \end{itemize}
\end{definition}

As a shorthand, we say that a nominal DFA $M$ has $n$ orbits or nominal dimension $k$ when the state set has $n$ orbits or nominal dimension $k$.

Given an input string $x = x_1 x_2 \cdots x_n \in A^*$, define the \emph{run} of $M$ on $x$ to be the sequence of states $\alpha_0, \ldots, \alpha_n \in Q$ such that $\alpha_0 = q_0$, and for $1 \leq i \leq n$, we have that $\delta(\alpha_{i-1}, x_i) = \alpha_i$. 
A string $x \in \Sigma^*$ is \emph{accepted} by $M$ if the last state appearing in the run of $M$ on $x$ is in $F$.
We say that a langauge $L$ is \emph{recognized} by a nominal DFA $M$ if $M$ accepts a string $x$ if and only if $x \in L$.
As noted in \cite[Definition 3.1]{bojanczyk-klin-lasota:automata-theory-in-nominal-sets}, the language recognized by a nominal DFA must be a $G$-language.
We say that a $G$-language $L$ is \emph{nominal regular} if it is recognized by some nominal DFA.

\begin{example}
    Let $A = \bb N$, and let $L : A \to \{0,1\}$ be defined by $L(w) = 1$ if and only if $w = aa$ for some $a \in A$ as in the previous example. 
    It is straightfoward to see this language is a $G$-language.
    Additionally, in the automaton that recognizes $L$, we can define an action of $G$ on the set of states by $\sigma \cdot q_i = q_{\sigma(i)}$ for any $\sigma \in G$ and $i \in \bb N$, while $\sigma \cdot q = q$ for all $\sigma \in G$ for the states $q_I, q_A,$ and $q_R$. 
    This automaton is a nominal DFA, and hence $L$ is nominal regular.
\end{example}

Let $L$ be a $G$-language, and define the usual relation $\equiv_L$ on $A^*$ by: $x \equiv_L y$ if and only if for all $z \in A^*$, we have $L(xz) = L(yz)$. 
This relation is equivariant (see \cite[Lemma 3.4]{bojanczyk-klin-lasota:automata-theory-in-nominal-sets}), and hence by \autoref{lem:quotient by equivariant equivalence relation is nominal}, $A^* / \equiv_L$ is a nominal set.
We will write $[x]_L$ to denote the $\equiv_L$-equivalence class of a string $x \in A^*$.

\begin{theorem}[Myhill-Nerode for nominal DFAs {\cite[Theorem 5.2]{bojanczyk-klin-lasota:automata-theory-in-nominal-sets}}] \label{thm:nominal myhill-nerode}
    Let $A$ be an orbit-finite nominal set, and let $L$ be a $G$-language.
    Then the following are equivalent:
    \begin{enumerate}
        \item $A^* / \equiv_L$ has at most $n$ orbits and has nominal dimension at most $k$;
        \item $L$ is nominal regular, and in particular is recognized by a nominal DFA with at most $n$ orbits and nominal dimension at most $k$.
    \end{enumerate}
\end{theorem}

Note that this statement is more precise than the original version in \cite{bojanczyk-klin-lasota:automata-theory-in-nominal-sets}; in particular, the conditions on the number of orbits and the nominal dimension do not appear in the original theorem.
However, these conditions are easily derived from its proof. 
\ifappendix For completeness, a proof of the correspondence between the number of orbits and the nominal dimension is given in Appendix \ref{appendix:proof of nominal myhill-nerode}. \fi

\subsection{Littlestone dimension of nominal DFAs}

To prove bounds on the Littlestone dimension of nominal automata, we wish to bound the number of possible nominal automata. Since automata involve transition functions, we will need to understand products of nominal sets.
Nominality is easily seen to be preserved under Cartesian products:

\begin{proposition}\label{prop:dimension of product}
    The product of two nominal sets is nominal. 
    In particular, the nominal dimension of the product of two nominal sets $X \times Y$ is at most the sum of the nominal dimensions of $X$ and $Y$.
\end{proposition}

\begin{proof}
    Suppose $X$ and $Y$ are nominal with nominal dimensions $k$ and $\ell$, respectively. 
    Let $(x,y) \in X \times Y$.
    Notice that $supp(x) \cup supp(y)$, which has size at most $k + \ell$, supports $(x,y) \in X \times Y$.
\end{proof}

On the other hand, in the most general setting, products of orbit-finite nominal sets need not be orbit-finite \cite[Example 2.5]{bojanczyk-klin-lasota:automata-theory-in-nominal-sets}.
However, in our context of $\Sym(\bb N)$, products of orbit-finite sets are known to be orbit-finite \cite[Section 10]{bojanczyk-klin-lasota:automata-theory-in-nominal-sets}. 
We will need an explicit bound on the number of orbits of products of orbit-finite sets, which is the purpose of the following discussion.

For a natural number $k$, define $\bb N^{(k)} := \{(a_1, \ldots, a_k) \mid a_i \neq a_j \text{ for } i \neq j\}$.
$\bb N^{(k)}$ is a single-orbit nominal set when equipped with the pointwise action of $G$.

\begin{lemma}[{\cite[Lemma 4.13]{bojanczyk-klin-lasota:automata-theory-in-nominal-sets}}] \label{lem:single orbit set is surjective image of N^(k)}
    Given a nominal set $X$ of nominal dimension $k$ that has exactly one orbit, there is an equivariant surjection $f_X : \bb N^{(k)} \to X$.
\end{lemma}

\begin{definition}[{cf. \cite[Section 3]{moerman-etal:learning-nominal-automata}}]
    Given $k_1, \ldots, k_n \in \bb N$, let $f_{\bb N}(k_1, \ldots, k_n)$ denote the number of orbits of $\bb N^{(k_1)} \times \cdots \times \bb N^{(k_n)}$. 
\end{definition}

\begin{proposition}[{cf. \cite[Section 3]{moerman-etal:learning-nominal-automata}}]\label{prop:number of orbits of product}
    Let $X_i$ be a nominal set with $\ell_i$ orbits and nominal dimension $k_i$ for $i=1, \ldots, n$.
    Let $X := X_1 \times \cdots \times X_n$ and $\ell = \ell_1 \cdots \ell_n$. Then $X$ has at most $\ell f_{\bb N}(k_1, \ldots, k_n)$ many orbits.
\end{proposition}

\ifappendix \autoref{prop:number of orbits of product} is stated but not proven in \cite{moerman-etal:learning-nominal-automata}. For completeness, we include a proof in Appendix \ref{appendix:proof of number of orbits of product}. \fi
In light of \autoref{prop:number of orbits of product}, we can find bounds on the number of orbits of a product of nominal sets by bounding the value of $f_{\bb N}$.
We do so for the case $n=2$:

\begin{proposition} \label{prop:bounds on fN(k1,k2)}
    Suppose (without loss of generality) that $k_1 \geq k_2$. 
    Then 
    \[
        \left(\frac{k_1}{e}\right)^{k_2} \leq \binom{k_1}{k_2} k_2! \leq f_{\bb N}(k_1, k_2) \leq (2k_1)^{k_2}.
    \]
    In particular, if $k_2$ is a constant, then $f_{\bb N}(k_1, k_2) = \Theta(k_1^{k_2})$.
\end{proposition}

\ifappendix The proof of \autoref{prop:bounds on fN(k1,k2)} is given in Appendix \ref{appendix:proof of bounds on fN(k1,k2)}. \fi
\autoref{prop:bounds on fN(k1,k2)} will help us to calculate an explicit upper bound on the Littlestone dimension of nominal automata later, but we can first use it to substantiate a comment from the introduction about previously known query bounds.

The bound in \cite[Corollary 1]{moerman-etal:learning-nominal-automata} involves a $f_{\bb N}(p(n+m), pn(k+k\log k+1))$ factor, where $n$ is the number of orbits of the state set of the target automaton, $k$ is the nominal dimension of the target automaton, $p$ is the nominal dimension of the alphabet, and $m$ is the length of the longest counterexample.
Notice that $f_{\bb N}$ is non-decreasing in all coordinates. 
Therefore, $f_{\bb N}(p(n+m), pn(k+k\log k+1))$ is lower bounded by both $f_{\bb N}(pn, pnk)$ and $f_{\bb N}(m, nk)$.
Applying the lower bounds from \autoref{prop:bounds on fN(k1,k2)} yields the following:
\begin{remark} \label{remark:factorial bounds in moerman etal}
    The bound on the (EQ+MQ)-query complexity of nominal automata given in \cite[Corollary 1]{moerman-etal:learning-nominal-automata} is at least $(k^n n!)^p$, and at least $\min\left(\left(\frac{nk}{e}\right)^m, \left(\frac{m}{e}\right)^{nk}\right)$.
\end{remark}

We are now at a point where we can calculate a bound on the number of possible nominal DFAs and hence bound the Littlestone dimension. 
First, we state a general result on the number of single-orbit nominal sets:

\begin{proposition} \label{prop:number of single orbit nominal sets}
    The number of distinct (up to isomorphism) single-orbit nominal sets of nominal dimension at most $k$ is at most $2^{O(k^2)}$.
\end{proposition}

\ifappendix The proof of \autoref{prop:number of single orbit nominal sets} is given in Appendix \ref{appendix:proof of number of single orbit nominal sets}. \fi
We note that it is a completely general result about nominal sets and may be of independent interest.
For us, it will allow us to bound the number of possible state sets.

\begin{lemma} \label{lem:number of nominal state sets}
    The number of possible state sets for a nominal DFA with $n$ orbits and nominal dimension $k$ is $2^{O(nk^2)}$.
\end{lemma}

\begin{proof}
    Since the state set $Q$ is an orbit-finite nominal set with $n$ orbits and nominal dimension $k$, we count the number of such sets. 
    We can view $Q$ as the disjoint union of $n$ single-orbit nominal sets, each with nominal dimension at most $k$, just by considering each orbit independently.
    By \autoref{prop:number of single orbit nominal sets}, the number of single-orbit nominal sets of nominal dimension $\leq k$ is at most $2^{O(k^2)}$.
    An upper bound on the number of nominal sets with $n$ orbits and nominal dimension $k$ can be found by raising this number to $n$, i.e., $\left(2^{O(k^2)}\right)^n = 2^{O(nk^2)}$.
\end{proof}

Now, we count the number of possible transition behaviors.
Fix an input alphabet $A$, where $A$ has $\ell$ orbits and nominal dimension $p$.

\begin{lemma} \label{lem:number of nominal transition functions}
    The number of possible transition behaviors for a nominal DFA with $n$ orbits and nominal dimension $k$ is at most 
    \[
    \left(n (k+p)!\right)^{O\left(nk^{p}\right)}.
    \]
\end{lemma}

The proof of \autoref{lem:number of nominal transition functions} is given in Appendix \ref{appendix:proof of number of nominal transition functions}

Let $\mc L^{\text{nom}}_A(n,k)$ denote the set of $G$-languages over $A$ recognized by a nominal DFA with at most $n$ orbits and nominal dimension at most $k$.

\begin{proposition} \label{prop:littlestone dimension of nominal automata}
    The Littlestone dimension of $\mc L^{\text{nom}}_A(n,k)$ is at most 
    \[
    O\left(nk^{p} \left(\log n + k \log k \right)\right)
    \]
\end{proposition}

\begin{proof}
    We bound the Littlestone dimension by bounding the size of $\mc L^{\text{nom}}_A(n,k)$. 
    To choose a nominal DFA with at most $n$ orbits and nominal dimension at most $k$, we must choose the state set, transition function, initial state, and set of accepting states. 
    By \autoref{lem:number of nominal state sets}, there are at most $2^{O(nk^2)}$ choices for the state set. 
    By \autoref{lem:number of nominal transition functions}, there are at most $\left(n (k+p)!\right)^{O(nk^{p})}$ choices for the transition function.
    There are at most $n$ choices for the initial state (since the initial state must be an orbit of its own) and $2^n$ choices for the accepting states (since the set of accepting states is a union of orbits).
    So in total, we can upper bound $|\mc L^{nom}_A(n,k)|$ by 
    \[
        |\mc L^{nom}_A(n,k)| \leq 2^{O\left(nk^2\right)} \left(n (k+p)!\right)^{O\left(nk^{p}\right)} n 2^n.
    \]
    Applying \autoref{remark:littlestone dimension of a finite class}, taking logs, using Stirling's approximation, and remembering that $p$ is constant, we find that 
    \begin{align*}
        \Ldim(\mc L^{nom}_A(n,k)) &\leq \log |\mc L^{nom}_A(n,k)| \\
        &\leq O\left(nk^2\right) + O\left(nk^{p}\right)(\log n + \log((k+p)!)) + \log n + n \\
        &\leq O\left(nk^{p} \left(\log n + (k+p)\log(k+p) \right)\right) \\
        &= O\left(nk^{p} \left(\log n + k \log k \right)\right).
    \end{align*}
\end{proof}

\subsection{Consistency dimension of nominal DFAs}

It remains to bound the consistency dimension.
To do this, we need to show that if $L$ is a language that is not recognized by a nominal DFA with at most $n$ orbits and nominal dimension at most $k$, then there is a set of strings $B$ of bounded size such that the restriction of $L$ to $B$ cannot be extended to any $G$-language recognized by such a nominal DFA. 
By the nominal Myhill-Nerode theorem, $L$ not being recognized by a nominal DFA with at most $n$ orbits and nominal dimension at most $k$ is equivalent to $A^* / \equiv_L$ either having greater than $n$ orbits or having nominal dimension greater than $k$.
We first address the case where $A^* / \equiv_L$ has large nominal dimension.

\begin{lemma} \label{lem:moving one element outside of least support moves x}
    Let $X$ be any nominal set, let $x \in X$, and let $D = supp(x)$. 
    Suppose that $\tau \in G$ fixes every element of $D$ except for one. 
    Then $\tau \cdot x \neq x$. 
\end{lemma}

\ifappendix The proof of this lemma is given in Appendix \ref{appendix:proof of moving one element outside of least support moves x}. \fi

\begin{proposition} \label{prop:high nominal dimension of L witnessed by small set}
    Let $L$ be a $G$-language over alphabet $A$, and suppose that $A^* / \equiv_L$ has nominal dimension at least $k+1$. 
    Then there is a set $B \subseteq A^*$ of size $2(k+1)$ such that for any $G$-language $L'$, if $L|_B = L'|_B$, then $A^* / \equiv_{L'}$ also has nominal dimension at least $k+1$. 
\end{proposition}

\begin{proof}
    Since $A^* / \equiv_L$ has nominal dimension $\geq k+1$, there is some $x_0 \in A^*$ such that $|supp([x_0]_L)| \geq k+1$. 
    Let $D$ denote $supp([x_0]_L)$, and let $D_0$ be a subset of $D$ of size $k+1$. 
    Let $j = \max(D)+1$, and for each $i \in D_0$, let $\tau_i = (i \enspace i+j) \in G$ (the permutation that swaps $i$ and $i+j$). 
    Notice that $\tau_i$ fixes all but one element of $D$, and so by \autoref{lem:moving one element outside of least support moves x}, $[\tau_i \cdot x_0]_L = \tau_i \cdot [x_0]_L \neq [x_0]_L$.  
    Thus for each $i \in D_0$, there is some $z_i \in A^*$ such that $L((\tau_i \cdot x_0) z_i) \neq L(x_0 z_i)$. 
    Let
    \[
        B = \{(\tau_i \cdot x_0) z_i \mid i \in D_0\} \cup \{x_0 z_i \mid i \in D_0\},
    \]
    which has size $2(k+1)$.
    
    Now, let $L'$ be a $G$-language extending $L|_B$, and suppose for contradiction that $A^* / \equiv_{L'}$ has nominal dimension at most $k$. 
    That is, for every $w \in A^*$, $|supp([w]_{L'})| \leq k$. 
    In particular, $|supp([x_0]_{L'})| \leq k$.
    For each $i \in D_0$, $L'$ agrees with $L$ on $(\tau_i \cdot x_0) z_i$ and $x_0z_i$, so $L'((\tau_i \cdot x_0) z_i) \neq L'(x_0 z_i)$. 
    This shows that $\tau_i \cdot [x_0]_{L'} = [\tau_i \cdot x_0]_{L'} \neq [x_0]_{L'}$. 
    Since the only elements not fixed by $\tau_i$ are $i$ and $i+j$, it must be that at least one of $i$ and $i+j$ are in the least support of $[x_0]_{L'}$. 
    Thus for each $i \in D_0$, $supp([x_0]_{L'})$ contains at least one of $i$ and $i+j$.
    Since $j > \max(D)$, all values $i, i+j$ for $i \in D_0$ are distinct, and so $|supp([x_0]_{L'})| \geq |D_0| = k+1$, which contradicts the fact that $|supp([x_0]_{L'})| \leq k$.
\end{proof}

Next, we address the case where $A^* / \equiv_L$ has a large number of orbits.

\begin{lemma} \label{lem:short witnesses to distinct equiv_L classes}
    Let $L$ be a $G$-language such that $A^* / \equiv_L$ has $n$ orbits. 
    Then for every orbit $G \cdot [x]_L$ of $A^* / \equiv_L$, there is a string $x' \in A^*$ such that $|x'| < n$ such that $[x']_L \in G \cdot [x]_L$. 
\end{lemma}

\ifappendix The proof of \autoref{lem:short witnesses to distinct equiv_L classes} is given in Appendix \ref{appendix:proof of short witnesses to distinct equiv_L classes}. \fi

\begin{proposition} \label{prop:many orbits of L witnessed by small set}
    Let $L$ be a $G$-language over alphabet $A$, where $A$ has nominal dimension $p$, and suppose that $A^* / \equiv_L$ has at least $n+1$ orbits. 
    Then there is a set $B \subseteq A^*$ of size $2 \binom{n+1}{2} \binom{pn}{k} (3pn)^k$ such that for any $G$-language $L'$, if $L|_B = L'|_B$, then $A^* / \equiv_{L'}$ also has at least $n+1$ orbits.
\end{proposition}

\begin{proof}
    Since $A^* / \equiv_L$ has at least $n+1$ orbits, there are strings $x_0, \ldots, x_n$ such that $[x_i]_{L}$ all belong to distinct orbits. 
    That is, for every $\tau \in G$ and $0 \leq i < j \leq n$, $[\tau \cdot x_i]_L = \tau [x_i]_L \neq [x_j]_L$, and thus there is $z^\tau_{ij} \in A^*$ such that 
    \[
        L((\tau \cdot x_i) z^\tau_{ij}) \neq L(x_j z^\tau_{ij}).
    \]
    Moreover, by \autoref{lem:short witnesses to distinct equiv_L classes}, we may assume that $|x_i| \leq n$ for each $0 \leq i \leq n$. 

    For each $0 \leq i \leq n$, let $D_i = supp(x_i)$. 
    Since the nominal dimension of $A$ is $p$, we have that $|D_i| \leq p \cdot |x_i| \leq pn$. 
    Also, let $D$ be some subset of $\bb N$ such that $D$ is disjoint from $D_0, \ldots, D_n$, and $\displaystyle |D| = \max_{0 \leq i \leq n} |D_i| \leq pn$.

    Now, for each $0 \leq i < j \leq n$, let
    \[
        \Sigma_{ij}' := \{\sigma' : D'' \to D_i \cup D_j \cup D \ \mid \ D'' \subseteq D_i \text{ of size } k \text{ and $\sigma'$ is an injection}\}.
    \]
    Notice that $\Sigma_{ij}'$ has size at most $\binom{pn}{k} (3pn)^k$: to choose a $\sigma'$, we choose a subset of $D_i$ of size $k$ and a value from $D_i \cup D_j \cup D$ for each of the $k$ inputs.
    For each $\sigma' \in \Sigma_{ij}'$, let $\sigma \in G$ be an arbitrary but fixed extension of $\sigma'$ to a permutation of $\bb N$, and let $\Sigma_{ij}$ consist of all such $\sigma$ (so $|\Sigma_{ij}| = |\Sigma_{ij}'| \leq \binom{pn}{k} (3pn)^k$). 

    Now, let 
    \begin{align*}
        B = \ &\{(\sigma \cdot x_i) z^\sigma_{ij} \ \mid \ 0 \leq i < j \leq n, \ \sigma \in \Sigma_{ij}\} \cup \\
        &\{x_j z^\sigma_{ij} \ \mid \ 0 \leq i < j \leq n, \ \sigma \in \Sigma_{ij}\}.
    \end{align*}
    $B$ has size at most $2 \binom{n+1}{2} \binom{pn}{k} (3pn)^k$. 
    Let $L'$ be a $G$-language extending $L|_B$, and assume for contradiction that $A^* / \equiv_{L'}$ has at most $n$ orbits. 
    Thus there must be $0 \leq i < j \leq n$ such that $[x_i]_{L'}$ is in the same orbit as $[x_j]_{L'}$. 
    Let $\pi \in G$ such that $\pi \cdot [x_i]_{L'} = [x_j]_{L'}$. 
    $\pi$ does not have to be in $\Sigma_{ij}$, and so we cannot directly derive a contradiction. 
    Instead, we will alter $\pi$ in order to obtain a $\sigma \in \Sigma_{ij}$ such that $\sigma \cdot [x_i]_{L'} = [x_j]_{L'}$.

    Define $\pi' \in G$ in the following way: first, let $\pi(D_i)$ denote the image of $D_i$ under $\pi$.
    For each element $a \in \pi(D_i) \setminus (D_i \cup D_j \cup D)$, select a unique element $b_a \in D \setminus \pi(D_i)$.
    This is possible because $|D| \geq |D_i| = |\pi(D_i)|$, and so $|D \setminus \pi(D_i)| \geq |\pi(D_i) \setminus D| \geq |\pi(D_i) \setminus (D_i \cup D_j \cup D)|$. 
    Let $\pi'$ be the permutation that transposes each $a \in \pi(D_i) \setminus (D_i \cup D_j \cup D)$ with its corresponding $b_a \in D \setminus \pi(D_i)$, and fixes every other element of $\bb N$. 

    Notice that $\pi'|_{D_j} = \mathrm{id}_{D_j}$, and so $\pi' \cdot x_j = x_j$. 
    Also, if $a \in D_i$, then $\pi' \circ \pi(a) \in D_i \cup D_j \cup D$: if $\pi(a) \in D_i \cup D_j \cup D$, then by definition $\pi'$ does not affect $\pi(a)$ and so $\pi'(\pi(a)) = \pi(a) \in D_i \cup D_j \cup D$, whereas if $\pi(a) \in \pi(D_i) \setminus (D_i \cup D_j \cup D)$, $\pi'$ will transpose $\pi(a)$ with an element in $D$. 
    Let $D_i'$ denote $supp([x_i]_{L'})$.
    Since $A^* / \equiv_{L'}$ has nominal dimension at most $k$, $|D_i'| \leq k$.
    Also, by \autoref{lem:dimension of quotient at most dimension of original}, $D_i' \subseteq supp(x_i) = D_i$.
    Therefore, $(\pi' \circ \pi)|_{D_i'}$ is an injection from a subset of $D_i$ of size at most $k$ into $D_i \cup D_j \cup D$. 
    Therefore, there is some $\sigma \in \Sigma_{ij}$ such that
    \[
        \sigma|_{D_i'} = (\pi' \circ \pi)|_{D_i'}.
    \]
    We may then deduce that 
    \begin{align*}
        [\sigma \cdot x_i]_{L'} &= \sigma \cdot [x_i]_{L'} \\
        &= (\pi' \circ \pi) \cdot [x_i]_{L'} \\
        &= \pi' \cdot (\pi \cdot [x_i]_{L'}) \\
        &= \pi' \cdot [x_j]_{L'} \\
        &= [\pi' \cdot x_j]_{L'} \\
        &= [x_j]_{L'}.
    \end{align*}

    This means that for all $z \in A^*$, $L'((\sigma \cdot x_i) z) = L'(x_j z)$.
    In particular, we may choose $z = z^\sigma_{ij}$.
    However, since $(\sigma \cdot x_i) z^\sigma_{ij}$ and $x_j z^\sigma_{ij}$ are in $B$, it must be that
    \begin{align*}
        L((\sigma \cdot x_i) z^\sigma_{ij}) &= L'((\sigma \cdot x_i) z^\sigma_{ij}) \\
        &= L'(x_j z^\sigma_{ij}) \\
        &= L(x_j z^\sigma_{ij}),
    \end{align*}
    a contradiction!
    So we may finally conclude that $A^* / \equiv_{L'}$ has at least $n+1$ orbits.    
\end{proof}

\autoref{lem:short witnesses to distinct equiv_L classes} yields an interesting consequence, which we do not use in any of our proofs but may be of independent interest:

\begin{corollary} \label{cor:dimension of A^*/equiv_L bounded in terms of orbits}
    If $L$ is a $G$-language over an alphabet $A$ which has nominal dimension $p$ such that $A^* / \equiv_L$ has $n$ orbits, then $A^* / \equiv_L$ has nominal dimension at most $(n-1)p$.
\end{corollary}

This result suggests that ultimately it is the number of orbits, and not the nominal dimension, of $A^* / \equiv_L$ which characterizes the complexity of $L$.
\ifappendix Its proof is given in Appendix \ref{appendix:proof of dimension of A^*/equiv_L bounded in terms of orbits}. \fi

\bigskip

We now have all the ingredients we need to bound the consistency dimension.

\begin{proposition} \label{prop:consistency dimension of nominal automata}
    The consistency dimension of $\mc L^{nom}_A(n,k)$ with respect to itself is at most $2 \binom{n+1}{2} \binom{pn}{k} (3pn)^k$.
\end{proposition}

\begin{proof}
    Let $L : A^* \to \{0,1\}$, and suppose that $L \notin \mc L^{nom}_A(n,k)$, i.e., $L$ is not recognized by any nominal DFA with at most $n$ orbits and nominal dimension at most $k$.
    We must find a set $B$ of at most $2 \binom{n+1}{2} \binom{pn}{k} (3pn)^k$ strings such that any function extending $L|_B$ is not in $\mc L^{nom}_A(n,k)$.

    \textbf{Case 1:} $L$ is not equivariant.
    Then there is $x_0 \in A^*$ and $\pi \in G$ such that $L(x_0) \neq L(\pi \cdot x_0)$.
    Set $B = \{x_0, \pi \cdot x_0\}$.
    Any function $L'$ extending $L|_B$ cannot be equivariant, and therefore cannot be nominal regular, so $L' \notin \mc L^{nom}_A(n,k)$. 

    \textbf{Case 2:} $L$ is equivariant, and is recognized by some nominal DFA with at most $n$ orbits. 
    By \autoref{thm:nominal myhill-nerode}, $A^* / \equiv_L$ has at most $n$ orbits. 
    Thus the nominal dimension of $A^* / \equiv_L$ must be at least $k+1$, or else another application of \autoref{thm:nominal myhill-nerode} would imply that $L$ is recognized by a nominal DFA with at most $n$ orbits and nominal dimension at most $k$, contradicting the assumption that $L \notin \mc L^{nom}_A(n,k)$. 
    Let $B \subseteq A^*$ of size $2(k+1)$ be as given by \autoref{prop:high nominal dimension of L witnessed by small set}.
    
    Now, let $L' : A^* \to \{0,1\}$ extend $L|_B$, and suppose for contradiction that $L' \in \mc L^{nom}_A(n,k)$. 
    Since $L'$ is recognized by a nominal DFA, it must be a $G$-language, so by the definition of $B$, $A^* / \equiv_{L'}$ has nominal dimension at least $k+1$.
    However, by \autoref{thm:nominal myhill-nerode} and the fact that $L'$ is recognized by a nominal DFA with at most $n$ orbits and nominal dimension at most $k$, $A^* / \equiv_{L'}$ has nominal dimension at most $k$, a contradiction!
    So $L' \notin \mc L^{nom}_A(n,k)$.

    \textbf{Case 3:} $L$ is equivariant, and $L$ is not recognized by any nominal DFAs with at most $n$ orbits.
    By \autoref{thm:nominal myhill-nerode}, $A^* / \equiv_L$ has at least $n+1$ orbits.
    Let $B \subseteq A^*$ of size $2 \binom{n+1}{2} \binom{pn}{k} (3pn)^k$ be as given by \autoref{prop:many orbits of L witnessed by small set}.

    Now, let $L' : A^* \to \{0,1\}$ extend $L|_B$, and suppose for contradiction that $L' \in \mc L^{nom}_A(n,k)$. 
    Since $L'$ is recognized by a nominal DFA, it must be a $G$-language, so by the definition of $B$, $A^* / \equiv_{L'}$ has at least $n+1$ orbits.
    However, by \autoref{thm:nominal myhill-nerode} and the fact that $L'$ is recognized by a nominal DFA with at most $n$ orbits and nominal dimension at most $k$, $A^* / \equiv_{L'}$ has at most $n$ orbits, a contradiction!
    So $L' \notin \mc L^{nom}_A(n,k)$.
    
    In all three cases, we found a set $B$ of size at most $2 \binom{n+1}{2} \binom{pn}{k} (3pn)^k$ such that any language extending $L|_B$ cannot be in $\mc L^{nom}_A(n,k)$. 
    We can then conclude that the consistency dimension of $\mc L^{nom}_A(n,k)$ with respect to itself is at most $2 \binom{n+1}{2} \binom{pn}{k} (3pn)^k$.
\end{proof}

\subsection{Learning bound for nominal DFAs}

We can now use the bounds we proved for Littlestone dimension and consistency dimension to prove our main result on nominal DFAs:

\nominalLCfixedalphabet

\begin{proof}
    By \autoref{prop:littlestone dimension of nominal automata}, the Littlestone dimension of $\mc L^{nom}_A(n,k)$ is at most $O\left(nk^{p} \left(\log n + k \log k \right)\right)$. 
    By \autoref{prop:consistency dimension of nominal automata}, the consistency dimension of $\mc L^{nom}_A(n,k)$ with respect to itself is at most $2 \binom{n+1}{2} \binom{pn}{k} (3pn)^k$.
    This is upper bounded by $n(n+1) \frac{(e \cdot pn)^k}{k^k} (3pn)^k$, which is in turn at most $\frac{n^{O(k)}}{k^k}$ (since $p$ is a constant).
    Applying \autoref{thm:query complexity bounds}, the query complexity of $\mc L^{nom}_A(n,k)$ with queries from $\mc L^{nom}_A(n,k)$ is at most $O\left(nk^{p} \left(\log n + k \log k \right) \frac{n^{O(k)}}{k^k}\right) = \frac{n^{O(k)}}{k^k}$.
\end{proof}

\section{Conclusion} \label{section:conclusion}

In this paper, we derive query learning bounds for advice DFAs and nominal DFAs; the bounds for advice DFAs are the first known bounds for the setting, while the bounds for nominal DFAs improve upon prior results, at the cost of not making any computational guarantees about the learning algorithm.

The bounds for advice DFAs are derived in almost the exact same manner as the bounds for DFAs given by Chase and Freitag \cite{chase-freitag:bounds-in-query-learning}. 
However, to the complexity of the setting of nominal DFAs requires a much deeper analysis of the structure of nominal sets and nominal DFAs.
This leads to several auxillary results which may be of independent interest; namely, \autoref{prop:bounds on fN(k1,k2)}, which analyzes the asymptotics of the number of orbits of products of nominal sets; \autoref{prop:number of single orbit nominal sets}, which counts the number of single-orbit nominal sets of a given dimension; and \autoref{cor:dimension of A^*/equiv_L bounded in terms of orbits}, which bounds the nominal dimension of $A^* / \equiv_L$ in terms of the number of orbits of $A^* / \equiv_L$ and the nominal dimension of $A$.

One natural extension of our work is to adapt the results for nominal automata to other well-behaved symmetries beside the equality symmetry, such as the total order symmetry, as described in \cite{bojanczyk-klin-lasota:automata-theory-in-nominal-sets}. 
Our approach is general enough that it should apply to these symmetries; however, there are many details that would need to be verified. 
It would also be natural to explore what the approach of Chase and Freitag yields for other previously-studied forms of automata that admit versions of the Myhill-Nerode theorem. 
For example, Bollig et al. \cite{bollig-etal:learning-NFA} study query learning of NFAs via \textit{residual finite-state automata}, which crucially utilize the Myhill-Nerode correspondence to find canonical NFA representations of regular languages. 
A final potential direction for future work would be to consider the generalization of nominal automata to $\omega$-languages; i.e., nominal DFAs with a B\"uchi-style acceptance criteria. 
Angluin and Fisman have adapted the $L^*$ algorithm to the classical setting of $\omega$-regular languages \cite{angluin-fisman:learning-regular-omega-languages}, and Chase and Freitag apply their method in the same setting \cite{chase-freitag:bounds-in-query-learning}.
To our knowledge, no work has been done studying the learnability of $\omega$-nominal regular languages, so it would be of interest to find results using both Angluin's method and Chase and Freitag's method. 

% \begin{credits}
% \subsubsection{\ackname} 
% \end{credits}
%

\bibliographystyle{splncs04}
\bibliography{khz-bib.bib}

\begin{thebibliography}{10}
\providecommand{\url}[1]{\texttt{#1}}
\providecommand{\urlprefix}{URL }
\providecommand{\doi}[1]{https://doi.org/#1}

\bibitem{angluin:learning-regular-sets}
Angluin, D.: Learning regular sets from queries and counterexamples. Inf.
  Comput.  \textbf{75},  87--106 (1987)

\bibitem{angluin-fisman:learning-regular-omega-languages}
Angluin, D., Fisman, D.: Learning regular omega languages. Theoretical Computer
  Science  \textbf{650},  57--72 (2016).
  \doi{https://doi.org/10.1016/j.tcs.2016.07.031},
  \url{https://www.sciencedirect.com/science/article/pii/S0304397516303760},
  algorithmic Learning Theory

\bibitem{balcazar-castro-guijarro-simon:consistency-dimension-query-learning}
Balcázar, J.L., Castro, J., Guijarro, D., Simon, H.U.: The consistency
  dimension and distribution-dependent learning from queries. Theoretical
  Computer Science  \textbf{288}(2),  197--215 (2002).
  \doi{https://doi.org/10.1016/S0304-3975(01)00400-5},
  \url{https://www.sciencedirect.com/science/article/pii/S0304397501004005},
  algorithmic Learning Theory

\bibitem{barany:automatic-omega-words-decidable-MSO-theory}
B\'ar\'any, V.: A hierarchy of automatic $\omega $-words having a decidable mso
  theory. RAIRO - Theoretical Informatics and Applications - Informatique
  Th\'eorique et Applications  \textbf{42}(3),  417--450 (2008).
  \doi{10.1051/ita:2008008},
  \url{http://www.numdam.org/articles/10.1051/ita:2008008/}

\bibitem{bojanczyk-klin-lasota:automata-theory-in-nominal-sets}
Bojańczyk, M., Klin, B., Lasota, S.: {Automata theory in nominal sets}.
  {Logical Methods in Computer Science}  \textbf{{Volume 10, Issue 3}} (Aug
  2014). \doi{10.2168/LMCS-10(3:4)2014},
  \url{https://lmcs.episciences.org/1157}

\bibitem{bollig-etal:learning-NFA}
Bollig, B., Habermehl, P., Kern, C., Leucker, M.: Angluin-style learning of
  nfa. In: IJCAI. pp. 1004--1009 (07 2009)

\bibitem{carton-thomas:monadic-theory-morphic-infinite-words}
Carton, O., Thomas, W.: The monadic theory of morphic infinite words and
  generalizations. In: Nielsen, M., Rovan, B. (eds.) Mathematical Foundations
  of Computer Science 2000. pp. 275--284. Springer Berlin Heidelberg, Berlin,
  Heidelberg (2000)

\bibitem{chase-freitag:bounds-in-query-learning}
Chase, H., Freitag, J.: Bounds in query learning. In: Abernethy, J., Agarwal,
  S. (eds.) Proceedings of Thirty Third Conference on Learning Theory.
  Proceedings of Machine Learning Research, vol.~125, pp. 1142--1160. PMLR
  (2020), \url{http://proceedings.mlr.press/v125/chase20a.html}

\bibitem{drews-dantoni:learning-symoblic-automata}
Drews, S., D'Antoni, L.: Learning symbolic automata. In: Legay, A., Margaria,
  T. (eds.) Tools and Algorithms for the Construction and Analysis of Systems.
  pp. 173--189. Springer Berlin Heidelberg, Berlin, Heidelberg (2017)

\bibitem{elgot-rabin:decidability-undecidability-extensions-of-successor}
Elgot, C.C., Rabin, M.: Decidability and undecidability of extensions of second
  (first) order theory of (generalized) successor. J. Symb. Log.  \textbf{31},
  169--181 (1966)

\bibitem{fisman-saadon:learning-characterizing-fully-ordered-lattice-automata}
Fisman, D., Saadon, S.: Learning and characterizing fully-ordered lattice
  automata. In: Bouajjani, A., Hol{\'i}k, L., Wu, Z. (eds.) Automated
  Technology for Verification and Analysis. pp. 266--282. Springer
  International Publishing, Cham (2022)

\bibitem{gabbay-pitts:new-approach-to-abstract-syntax}
Gabbay, M.J., Pitts, A.M.: A new approach to abstract syntax with variable
  binding. Formal Aspects of Computing  \textbf{13}(3),  341--363 (07 2002).
  \doi{10.1007/s001650200016}, \url{https://doi.org/10.1007/s001650200016}

\bibitem{hanneke-livni-moran:online-learning-simple-predictors}
Hanneke, S., Livni, R., Moran, S.: Online learning with simple predictors and a
  combinatorial characterization of minimax in 0/1 games. In: COLT (2021)

\bibitem{hellerstein-pillaipakkamnatt-raghavan-wilkins:how-many-queries-are-needed-to-learn}
Hellerstein, L., Pillaipakkamnatt, K., Raghavan, V., Wilkins, D.: How many
  queries are needed to learn? J. ACM  \textbf{43}(5),  840–862 (Sep 1996).
  \doi{10.1145/234752.234755}, \url{https://doi.org/10.1145/234752.234755}

\bibitem{kruckman-etal:myhill-nerode-automata-advice}
Kruckman, A., Rubin, S., Sheridan, J., Zax, B.: A myhill-nerode theorem for
  automata with advice. Electronic Proceedings in Theoretical Computer Science
  \textbf{96} (10 2012). \doi{10.4204/EPTCS.96.18}

\bibitem{moerman-etal:learning-nominal-automata}
Moerman, J., Sammartino, M., Silva, A., Klin, B., Szynwelski, M.: Learning
  nominal automata. In: Proceedings of the 44th ACM SIGPLAN Symposium on
  Principles of Programming Languages. p. 613–625. POPL 2017, Association for
  Computing Machinery, New York, NY, USA (2017). \doi{10.1145/3009837.3009879},
  \url{https://doi.org/10.1145/3009837.3009879}

\bibitem{montanari-pistore:history-dependent-automata-an-introduction}
Montanari, U., Pistore, M.: An introduction to history dependent automata.
  Electr. Notes Theor. Comput. Sci.  \textbf{10},  170--188 (01 1997).
  \doi{10.1016/S1571-0661(05)80696-6}

\bibitem{nies:describing-groups}
Nies, A.: Describing groups. The Bulletin of Symbolic Logic  \textbf{13}(3),
  305--339 (2007), \url{http://www.jstor.org/stable/4493323}

\bibitem{pitts:nominal-sets-names-and-symmetry-in-computer-science}
Pitts, A.M.: Nominal Sets: Names and Symmetry in Computer Science. Cambridge
  Tracts in Theoretical Computer Science, Cambridge University Press (2013).
  \doi{10.1017/CBO9781139084673}

\bibitem{pyber:asymptotic-results-permutation-groups}
Pyber, L.: Asymptotic results for permutation groups. In: Groups And
  Computation (1991)

\bibitem{rabinovich-thomas:decidable-theories-of-natural-numbers-unary-predicates}
Rabinovich, A., Thomas, W.: Decidable theories of the ordering of natural
  numbers with unary predicates. In: {\'E}sik, Z. (ed.) Computer Science Logic.
  pp. 562--574. Springer Berlin Heidelberg, Berlin, Heidelberg (2006)

\bibitem{yasubumi:learning-context-free-grammars}
Sakakibara, Y.: Learning context-free grammars from structural data in
  polynomial time. In: Proceedings of the First Annual Workshop on
  Computational Learning Theory. p. 330–344. COLT '88, Morgan Kaufmann
  Publishers Inc., San Francisco, CA, USA (1988)

\bibitem{salomaa:finite-automata-time-variant-structure}
Salomaa, A.: On finite automata with a time-variant structure. Information and
  Control  \textbf{13}(2),  85--98 (1968).
  \doi{https://doi.org/10.1016/S0019-9958(68)90706-7},
  \url{https://www.sciencedirect.com/science/article/pii/S0019995868907067}

\bibitem{tsankov:additive-group-of-rationals-not-automatic}
Tsankov, T.: The additive group of the rationals does not have an automatic
  presentation. The Journal of Symbolic Logic  \textbf{76}(4),  1341--1351
  (2011), \url{http://www.jstor.org/stable/23208221}

\end{thebibliography}

\ifappendix

\newpage
\appendix

\section{Omitted proofs from \autoref{section:complexity of advice DFAs}} \label{appendix:advice DFAs omitted proofs}

\subsection{A witness to the tightness of the bounds in \autoref{thm:advice myhill-nerode with bounds}} \label{appendix:witness to tightness of advice myhill-nerode bounds}

\begin{example}
    Define $L : \{0,1\}^* \to \{0,1\}$ by 
    \[
      L(w) = \begin{cases} 1 & (w \text{ has an even number of 0's} \land |w| \neq 2) \lor (|w| = 3) \\
                           0 & \text{otherwise} \end{cases}
    \]
  \end{example}

  Notice that $\equiv_{L,m}$ has at most two classes for every $m$---one for strings with an even number of 0's and one for strings with an odd number of 0's.
  However, it is not recognized by any DFA $M$ with advice with 3 states. 
  To see this, let $M$ accept $L$.
  The runs of strings 00 and 01 both end in reject states of $M$, but since they are in separate $\equiv_{L,2}$-classes, they must end in distinct states. 
  So there are at least two reject states.
  Similarly, the runs of 000 and 001 both end in distinct accept states. 
  Thus there must be at least 4 states in $M$.
  
  Replacing ``$w$ has an even number of 0's'' with any regular language that has at most $k$ equivalence classes in the Myhill-Nerode relation (for example, ``the number of 0's in $w$ is divisible by $k$''), we obtain a language $L_k$ such that $\equiv_{L_k,m}$ has at most $k$ classes for every $m$, but any advice DFA accepting $L_k$ must have at least $2k$ states.

  \section{Omitted proofs from \autoref{section:complexity of nominal DFAs}} \label{appendix:nominal DFAs omitted proofs}

  \subsection{Proof of \autoref{lem:induced equivariant equivalence relations}} \label{appendix:proof of induced equivariant equivalence relations}

  \begin{proof}
    Define $\equiv_X$ by $x_1 \equiv_X x_2$ if and only if $F(x_1) \equiv_Y F(x_2)$.
    A standard argument confirms that this is indeed an equivalence relation, equivariance follows from the fact that $F$ is equivariant.

    Next, define $f : X / \equiv_X \to \equiv_Y$ by $f([x]_{\equiv_X}) = [F(x)]_{\equiv_Y}$.
    This is well-defined: if $x_1 \equiv_X x_2$, then 
    \begin{align*}
        f([x_1]_{\equiv_X}) &= [F(x_1)]_{\equiv_Y} \\
        &= [F(x_2)]_{\equiv_Y} & \text{since } x_1 \equiv_X x_2 \Rightarrow F(x_1) \equiv_Y F(x_2) \\
        &= f([x_2]_{\equiv_X})
    \end{align*}
    $f$ is equivariant since if $\pi \in G$, then $f(\pi \cdot [x]_{\equiv_X}) = f([\pi \cdot x]_{\equiv_X}) = [F(\pi \cdot x)]_{\equiv_Y} = [\pi \cdot F(x)]_{\equiv_Y} = \pi \cdot [F(x)]_{\equiv_Y} = \pi \cdot f([x]_{\equiv_X})$.
    $f$ is injective since for $x_1, x_2 \in X$,
    \begin{align*}
        f([x_1]_{\equiv_X}) &= f([x_2]_{\equiv_X}) \\
        \Rightarrow [F(x_1)]_{\equiv_Y} &= [F(x_2)]_{\equiv_Y} \\
        \Rightarrow F(x_1) &\equiv_Y F(x_2) \\
        \Rightarrow x_1 &\equiv_X x_2 \\
        \Rightarrow [x_1]_{\equiv_X} &= [x_2]_{\equiv_X}.
    \end{align*}
    Finally, if $F$ is surjective, then given $[y]_{\equiv_Y} \in Y / \equiv_Y$, there is $x \in X$ such that $F(x) = y$ and hence $f([x]_{\equiv_X}) = [y]_{\equiv_Y}$, so $f$ is surjective.
\end{proof}

\subsection{Proof of \autoref{thm:nominal myhill-nerode}} \label{appendix:proof of nominal myhill-nerode}

\begin{definition}[Reachable Nominal DFA]
    A nominal DFA $M$ is said to be \emph{reachable} if for every state $q$ in $M$, there is $x \in A^*$ such that the run of $M$ on $x$ ends in state $q$.
\end{definition}

\begin{definition}[Syntactic Automaton]
    Fix an orbit-finite nominal alphabet $A$, and let $L : A^* \to \{0,1\}$ be a $G$-language. The \emph{syntactic automaton} of $L$, denoted $M_L$ is specified as follows: 
    \begin{itemize}
        \item the state set is the set $A^* / \equiv_L$;
        \item the transition function is $\delta_L : A^* / \equiv_L \times \ A \to A^* / \equiv_L$\\ defined by 
        \[
            \delta_L([x]_L, a) = [xa]_L,
        \]
        \item the initial state is $[\epsilon]_L$;
        \item the set of accepting states is $\{[x]_L \mid x \in L\}$.
    \end{itemize}
\end{definition}

\begin{lemma}[{\cite[Lemma 3.6; Proposition 5.1]{bojanczyk-klin-lasota:automata-theory-in-nominal-sets}}]
    The syntactic automaton of a $G$-language is a reachable nominal DFA.
\end{lemma}

By \cite[Lemma 3.7]{bojanczyk-klin-lasota:automata-theory-in-nominal-sets}, a $G$-language is always recognized by its syntactic automaton.

\begin{definition}[Automaton Homomorphism]
    Let $M = (Q, \delta, q_0, F)$ and $M' = (Q', \delta', q_0', F')$ be two nominal DFAs over the same alphabet $A$. 
    An \emph{automaton homomorphism} from $M$ to $M'$ is an equivariant function $f : Q \to Q'$ such that:
    \begin{itemize}
        \item $f(q_0) = q_0'$;
        \item $q \in F \iff f(q) \in F'$ for every $q \in Q$; and 
        \item $f(\delta(q,a)) = \delta'(f(q), a)$ for every $q \in Q, a \in A$.
    \end{itemize}
\end{definition}

If there exists an automaton homomorphism from $M$ to $M'$, then $M$ and $M'$ recognize the same language: for any string $x$, the run of $M$ on $x$ ends in state $q$ if and only if the run of $M'$ on $x$ ends in the state $f(q)$, and $q \in F$ if and only if $f(q) \in F'$, so $M$ accepts $x$ if and only if $M'$ accepts $x$.

\begin{lemma}[{\cite[Lemma 3.7]{bojanczyk-klin-lasota:automata-theory-in-nominal-sets}}] \label{lem:syntactic automaton is homomorphic image}
    Let $L$ be a $G$-language.
    For any reachable nominal DFA $M$ that recognizes $L$, there is a surjective automaton homomorphism $f$ from $M$ to $M_L$.
\end{lemma}

We can now prove bounds in the nominal Myhill-Nerode theorem:

\begin{proof}[Proof of \autoref{thm:nominal myhill-nerode}]
    (1. $\Rightarrow$ 2.) Suppose that $A^* / \equiv_L$ has at most $n$ orbits and nominal dimension at most $k$.
    Then $M_L$ is a nominal DFA that recognizes $L$, and its state set is $A^* / \equiv_L$ which by assumption has at most $n$ orbits and nominal dimension at most $k$.

    (2. $\Rightarrow$ 1.) Suppose that $L$ is recognized by a nominal DFA $M$ with at most $n$ orbits and nominal dimension at most $k$. 
    We may assume that $M$ is reachable, so by \autoref{lem:syntactic automaton is homomorphic image}, there is a surjective automaton homomorphism $f$ from $M$ to the syntactic automaton $M_L$. 
    By \autoref{lem:orbit of image under equivariant function} and \autoref{lem:supports preserved under equivariant functions}, the state set of $M_L$ has at most $n$ orbits and nominal dimension $k$. 
    But the state set of $M_L$ is exactly $A^* / \equiv_L$, and so we obtain the desired bounds.
\end{proof}

\subsection{Proof of \autoref{prop:number of orbits of product}} \label{appendix:proof of number of orbits of product}

\begin{proof}
    Let $(x_1, \ldots, x_n) \in X$.
    Its orbit $G \cdot (x_1, \ldots, x_n)$ must be contained in the product $(G \cdot x_1) \times \cdots \times (G \cdot x_n)$.
    Hence, it is enough to bound the number of orbits of any given product $O_1 \times \ldots \times O_n$, where each $O_i$ is an orbit of $X_i$, and multiply by the number of possible products.
    Fix orbits $O_1, \ldots, O_n$ of $X_1, \ldots, X_n$, respectively. 
    By \autoref{lem:single orbit set is surjective image of N^(k)}, there are equivariant surjections $f_{O_i} : \bb N^{(k_i)} \to O_i$. 
    Taking the product gives us the equivariant surjection
    \[
        \bb N^{(k_1)} \times \cdots \times \bb N^{(k_n)} \xrightarrow{f_{O_1} \times \cdots \times f_{O_n}} O_1 \times \cdots \times O_n.
    \]
    By \autoref{lem:orbit of image under equivariant function}, $O_1 \times \cdots \times O_n$ has at most as many orbits as $\bb N^{(k_1)} \times \cdots \times \bb N^{(k_n)}$, which has $f_{\bb N}(k_1, \ldots, k_n)$ orbits.
    Now, there were $\ell = \ell_1 \ell_2 \cdots \ell_n$ ways to choose the orbits $O_1, \ldots, O_n$, so in total, $X$ can have at most $\ell f_{\bb N}(k_1, \ldots, k_n)$ orbits.
\end{proof}

\subsection{Proof of \autoref{prop:bounds on fN(k1,k2)}} \label{appendix:proof of bounds on fN(k1,k2)}

\begin{proof}
    Consider tuples $\bar a = (a_1, \ldots, a_{k_1}) \in \bb N^{(k_1)}$ and $\bar b = (b_1, \ldots, b_{k_2}) \in \bb N^{(k_2)}$. 
    Notice that the orbit of the pair $(\bar a, \bar b)$ is exactly determined by the collection of indices $i,j$ such that $a_i = b_j$.
    So to choose an orbit, we can first choose the number of indices $0 \leq r \leq k_2$ that $\bar a$ and $\bar b$ coincide on.
    Then we need to choose $r$ indices $i_1, \ldots, i_r$ of $\bar a$, $r$ indices of $j_1, \ldots, j_r$ of $\bar b$, and a bijection between $\{i_1, \ldots, i_r\}$ and $\{j_1, \ldots, j_r\}$ in order to determine the indices that $\bar a$ and $\bar b$ coincide on. 
    There are $\binom{k_1}{r}\binom{k_2}{r}r!$ ways to choose these. 
    
    Thus the total number of orbits is 
    \[
    f_{\bb N}(k_1, k_2) = \sum_{r=0}^{k_2} \binom{k_1}{r}\binom{k_2}{r}r!,
    \]

    This is lower bounded by $\binom{k_1}{k_2} k_2!$.
    Since $k_2! \geq \left(\frac{k_2}{e}\right)^{k_2}$ and $\binom{k_1}{k_2} \geq \left(\frac{k_1}{k_2}\right)^{k_2}$ for any value of $k_1, k_2$, we have that
    \[
        f_{\bb N}(k_1, k_2) \geq \binom{k_1}{k_2} k_2! \geq \left(\frac{k_1}{k_2}\right)^{k_2} \left(\frac{k_2}{e}\right)^{k_2} = \left(\frac{k_1}{e}\right)^{k_2}.
    \]

    On the other hand, since $\binom{n}{k} \leq \frac{n^k}{k!}$ for any $n$ and $k$,
    \begin{align*}
        \sum_{r=0}^{k_2} \binom{k_1}{r}\binom{k_2}{r}r! &\leq \sum_{r=0}^{k_2} \frac{k_1^r}{r!} \binom{k_2}{r} r! \\
        &= \sum_{r=0}^{k_2} k_1^r \binom{k_2}{r} \\
        &\leq \sum_{r=0}^{k_2} k_1^{k_2} \binom{k_2}{r} \\
        &= k_1^{k_2} \sum_{r=0}^{k_2} \binom{k_2}{r} \\
        &= k_1^{k_2} 2^{k_2} = (2k_1)^{k_2}
    \end{align*}
\end{proof}

\subsection{Proof of \autoref{prop:number of single orbit nominal sets}} \label{appendix:proof of number of single orbit nominal sets}

The following definitions and propositions are adapted from \cite[Section 9]{bojanczyk-klin-lasota:automata-theory-in-nominal-sets}, and will allow us to bound the number of possible nominal sets.

\begin{definition}[{cf. \cite[Definition 9.11]{bojanczyk-klin-lasota:automata-theory-in-nominal-sets}}]
    A \emph{support representation} is a pair $(k,S)$, where $k \in \bb N$, and $S$ is a subgroup of $\Sym([k])$. 
\end{definition}

\begin{definition}[{cf. \cite[Definition 9.14]{bojanczyk-klin-lasota:automata-theory-in-nominal-sets}}] \label{def:G-set semantics}
    Given a support representation $(k,S)$, the \emph{semantics} of $(k,S)$, denoted $[k,S]^{ec}$, is the set $\bb N^{(k)} / \equiv_S$, where $\equiv_S$ is defined as
    \[
    (a_1, \ldots, a_k) \equiv_S (b_1, \ldots, b_k) \iff \\ \exists \tau \in S \ \forall i \in [k], \ a_{\tau(i)} = b_i .
    \]
    There is a natural action of $G$ on $[k,S]^{ec}$ defined by 
    \[
    \pi \cdot [(a_1, \ldots, a_k)]_S = [(\pi(a_1), \ldots, \pi(a_k))]_S.
    \]
\end{definition}

An \emph{isomorphism} of nominal sets is an equivariant bijection between two nominal sets.

\begin{proposition}[{cf. \cite[Proposition 9.15]{bojanczyk-klin-lasota:automata-theory-in-nominal-sets}}] \label{prop:characterization of single-orbit nominal sets}
    For any support representation $(k,S)$, $[k,S]^{ec}$ is a single-orbit nominal set of nominal dimension $k$, and every single-orbit nominal set $X$ of nominal dimension $k$ is isomorphic to $[k,S]^{ec}$ for some $S \leq \Sym([k])$. 
\end{proposition}

\begin{proposition}[{cf. \cite[Proposition 9.16]{bojanczyk-klin-lasota:automata-theory-in-nominal-sets}}] \label{prop:characterization of equivariant functions}
    Let $X = [k,S]^{ec}$ and $Y = [\ell,T]^{ec}$ be single-orbit nominal sets. 
    Let 
    \begin{align*}
        U = \{u : \ &[\ell] \to [k] \mid u \text{ is injective and } \\ 
        &\forall \sigma \in S  \ \exists \tau \in T, \ \sigma \circ u = u \circ \tau\}.
    \end{align*}
    Equivariant functions from $X$ to $Y$ are in bijective correspondence with $U / \equiv_T$ (where $\equiv_T$ is as in \autoref{def:G-set semantics}).
\end{proposition}

\begin{lemma} \label{lem:characterization of G-set semantics}
    $[k,S]^{ec}$ is determined, up to isomorphism, by $k$ and the conjugacy class of $S$ in $\Sym([k])$. 
\end{lemma}

\begin{proof}
    Let $X = [k,S]^{ec}$ and $Y = [\ell,T]^{ec}$ be single-orbit nominal sets.
    Suppose that $k = \ell$ and that $S,T$ are conjugate in $\Sym([k])$. 
    That is, there is a permutation $\rho : [k] \to [k]$ such that $\rho S \rho^{-1} = T$. 
    Define $F^\rho : \bb N^{(k)} \to \bb N^{(k)}$ by $F^\rho((a_1, \ldots, a_k)) = (a_{\rho(1)}, \ldots, a_{\rho(k)})$ (i.e., reorder the input using $\rho$). 
    Notice that $F^\rho$ is an equivariant bijection.
    By \autoref{lem:induced equivariant equivalence relations}, $F^\rho$ and $\equiv_T$ induce an equivalence relation on $\bb N^{(k)}$.
    Since $\rho S \rho^{-1} = T$, the induced equivalence relation is actually $\equiv_S$. 
    Then the induced function $f$ is an equivariant bijection between $X$ and $Y$.
    
    In the other direction, suppose there is an isomorphism $f : X \to Y$.
    The proof of \autoref{prop:characterization of equivariant functions} gives us a bijection $u : [\ell] \to [k]$ such that $uS = Tu$. 
    In particular, $k=\ell$ and $u$ is a permutation in $\Sym([k])$ that witnesses that $S$ and $T$ are conjugate. 
\end{proof}

With this machinery in hand, we can give the proof of \autoref{prop:number of single orbit nominal sets}:

\begin{proof}
    Let $X = [k',S]^{ec}$ be a single-orbit nominal set with nominal dimension at most $k$. 
    By \autoref{lem:characterization of G-set semantics}, $X$ is determined up to isomorphism by the value of $k'$ and the conjugacy class of $S$ in $\Sym([k'])$. 
    Since the nominal dimension of $X$ is $k$, $k' \leq k$, and so there are $k$ choices for $k'$. It remains to count the number of conjugacy classes of subgroups of $\Sym([k])$. 
    The number of subgroups of $\Sym([k])$ is $2^{\Theta(k^2)}$ \cite[Theorem 4.2]{pyber:asymptotic-results-permutation-groups}, which certainly gives an upper bound for the number of conjugacy classes of subgroups. 
    Additionally, any subgroup has at most $k!$ conjugates (one for each permutation in $\Sym([k])$), so the number of conjugacy classes is also at least $2^{O(k^2)}/k! = 2^{O(k^2)}$.
    
    Thus the number of single-orbit nominal sets of nominal dimension at most $k$ is at most $k 2^{O(k^2)} = 2^{O(k^2)}$. 
\end{proof}

\subsection{Proof of \autoref{lem:number of nominal transition functions}} \label{appendix:proof of number of nominal transition functions}

\begin{proof}
    We count the number of equivariant functions $\delta : Q \times A \to Q$, where $Q$ has $n$ orbits and nominal dimension $k$. 
    Since $\delta$ is equivariant, by \autoref{lem:orbit of image under equivariant function}, a single orbit of $Q \times A$ must map into a single orbit of $Q$, so we can first choose a target orbit of $Q$ for each orbit of $Q \times A$.
    By \autoref{prop:number of orbits of product}, $Q \times A$ has at most $n \ell f_{\bb N}(k,p)$ orbits.
    Since $p$ is fixed, by \autoref{prop:bounds on fN(k1,k2)}, for any $k$, $f_{\bb N}(k,p) = O(k^{p})$.
    Furthermore, since $\ell$ is also fixed, $Q \times A$ has at most $O(nk^{p})$ orbits.
    Thus there are $n^{O(nk^{p})}$ many ways to choose the target orbits of each orbit of $Q \times A$. 
    
    Once we have chosen a target orbit of $Q$ for each orbit of $Q \times A$, we must choose an equivariant function from each orbit $O_1$ of $Q \times A$ to one orbit $O_2$ of $Q$. 
    By \autoref{prop:dimension of product}, the nominal dimension of $Q \times A$ is at most $k+p$, and so $O_1$ also has nominal dimension at most $k+p$.
    Similarly, $O_2$ has nominal dimension at most $k$. 
    Therefore $O_1 = [k',S]^{ec}$ where $k' \leq k+p$, and $O_2 = [\ell,T]^{ec}$ where $\ell \leq k$.
    By \autoref{prop:characterization of equivariant functions}, the number of equivariant functions from $O_1$ to $O_2$ is upper bounded by the number of injections $[\ell] \to [k]$, which is in turn at most $\frac{(k+p)!}{p!} \leq (k+p)!$.
    
    We must choose one such equivariant function for each orbit of $Q \times A$, so the total number of choices of all of these functions is at most $\left((k+p)!\right)^{O(nk^{p})}$. 
    Once we have done this, we have chosen a transition behavior $\delta : Q \times A \to Q$.
    
    Hence the total number of possible transition behaviors is upper bounded by 
    \[
    n^{O(nk^{p})} \cdot \left((k+p)!\right)^{nk^{p}}= \left(n(k+p)!\right)^{O(nk^{p})}.
    \]
\end{proof}

\subsection{Proof of \autoref{lem:moving one element outside of least support moves x}} \label{appendix:proof of moving one element outside of least support moves x}

\begin{proof}
    Suppose for contradiction that $\tau \cdot x = x$. 
    Let $i \in D$ be the only element of $D$ such that $\tau(i) \neq i$, and notice that $\tau(i) \notin D$ since $\tau$ fixes every element of $D$ other than $i$.
    We claim that $D \setminus \{i\}$ supports $x$. 
    To see this, let $\sigma \in G$ such that $\sigma$ fixes every element of $D \setminus i$. 
    We need to show that $\sigma \cdot x = x$.
    We may assume that $\sigma(i) \neq i$, since otherwise $\sigma$ fixes every element of $D$, and since $D$ supports $x$, we would have $\sigma \cdot x = x$. 
    Additionally, $\sigma(i) \notin D$ since $\sigma$ fixes every element of $D$ other than $i$. 
    Now, let $j = \tau(i)$ and $k = \sigma(i)$, and consider the permutation $(j \ k) \sigma$. 
    We have that $(j \ k) \sigma|_D = \tau|_D$, as $\sigma$ and $\tau$ both fix every element of $D \setminus \{i\}$, and $\tau(i) = j = (j \ k)(k) = (j / k)\sigma(i)$. 
    Since $D$ supports $x$, we can deduce that $(j \ k) \sigma \cdot x = \tau \cdot x = x$. 
    Then $\sigma \cdot x = (j \ k) \cdot x$. 
    Since $j,k \notin D$, $(j \ k)$ fixes every element of $D$, and so $(j \ k) \cdot x = x$. 
    This shows that $\sigma \cdot x = x$, and since $\sigma$ was an arbitrary permutation that fixed every element of $D \setminus \{i\}$, we may conclude that $D \setminus \{i\}$ supports $x$.
    This is a contradiction, since $D$ is the \textit{least} support of $x$, and hence $\tau \cdot x \neq x$.
\end{proof}

\subsection{Proof of \autoref{lem:short witnesses to distinct equiv_L classes}} \label{appendix:proof of short witnesses to distinct equiv_L classes}

\begin{proof}
    Suppose for contradiction that there is some orbit $G \cdot [x]_L$ of $A^* / \equiv_L$ such that the orbit of every string of length strictly less than $n$ is not in $G \cdot [x]_L$. 
    We may assume that the representative of the orbit $x$ is of minimal length $m \geq n$. 
    Write $x = a_1 \cdots a_m$. 

    Consider the collection of prefixes of $x$ of length up to $n-1$, i.e., the set $\{\epsilon, a_1, a_1 a_2, \ldots, a_1 \cdots a_{n-1}\}$.
    Since these all have length strictly less than $n$, they must belong to one of the $(n-1)$-many orbits that are not in $G \cdot [x]_L$. 
    As there are $n$ strings in the collection, there must be $0 \leq i < j \leq n-1$ such that $[a_1 \cdots a_i]_L$ and $[a_1 \cdots a_j]_L$ belong to the same orbit (here, $a_1 \cdots a_0$ denotes the empty string). 
    That is, there is $\tau \in G$ such that $\tau \cdot (a_1 \cdots a_i) \equiv_L a_1 \cdots a_j$. 
    Notice that $\equiv_L$ is preserved under appending common suffixes; i.e., if $x \equiv_L y$, then for any $z \in A^*$, $xz \equiv_L yz$. Thus $[\tau \cdot (a_1 \cdots a_i)] a_{j+1} \cdots a_m \equiv_L a_1 \cdots a_j a_{j+1} \cdots a_m = x$. 
    However, notice that $[\tau \cdot (a_1 \cdots a_i)] a_{j+1} \cdots a_m$ is strictly shorter than $x$, and its $\equiv_L$-class is in the same orbit of $A^* / \equiv_L$ as $[x]_L$ (more than that, is in the same $\equiv_L$-class), which contradicts the assumption that $x$ was of minimal length.

    Thus for every orbit of $A^* / \equiv_L$, there must be a string of length strictly less than $n$ that is in the orbit.
\end{proof}

\subsection{Proof of \autoref{cor:dimension of A^*/equiv_L bounded in terms of orbits}} \label{appendix:proof of dimension of A^*/equiv_L bounded in terms of orbits}

\begin{proof}
    Let $[x]_L \in A^* / \equiv_L$. 
    By \autoref{lem:short witnesses to distinct equiv_L classes}, there is some $x' \in A^*$ with $|x'| < n$ and $\pi \in G$ such that $\pi \cdot [x]_L = [x']_L$. 
    Then \autoref{lem:size of support preserved under action} and \autoref{lem:dimension of quotient at most dimension of original} tell us that 
    \[
        |supp([x]_L)| = |supp(\pi \cdot [x]_L)| = |supp([x']_L)| \leq |supp(x')|.
    \]
    Each of the letters of $x'$ is supported by a set of size at most $p$.
    Since $G$ acts on $x'$ coordinate-wise, the least support of $x'$ is contained in the union of the least supports of all the letters of $x'$, which is a set of size at most $(n-1)p$.
    Hence every element of $A^* / \equiv_L$ is supported by a set of size at most $(n-1)p$, and so the nominal dimension of $A^* / \equiv_L$ is at most $(n-1)p$.
\end{proof}

\fi

\end{document}